%% file: NoisyAdaptiveGT.tex
\documentclass[english,onecolumn,draftcls]{IEEEtran}
 
\makeatletter

\input{preamble.tex}

\usepackage{hyperref} 
\usepackage{enumitem}
\usepackage{float}

\floatstyle{ruled}
\newfloat{algorithm}{tbp}{loa}
\providecommand{\algorithmname}{Algorithm}
\floatname{algorithm}{\protect\algorithmname}

\setenumerate[1]{label=\arabic*.}

\usepackage{geometry} 
\makeatletter
\newcommand{\manuallabel}[2]{\def\@currentlabel{#2}\label{#1}}
\makeatother

\begin{document} 

\title{Noisy Adaptive Group Testing: \\ Bounds and Algorithms}
\author{Jonathan Scarlett}
\maketitle

\begin{abstract}
    The group testing problem consists of determining a small set of defective items from a larger set of items based on a number of possibly-noisy tests, and is relevant in applications such as medical testing, communication protocols, pattern matching, and many more.  One of the defining features of the group testing problem is the distinction between the non-adaptive and adaptive settings: In the non-adaptive case, all tests must be designed in advance, whereas in the adaptive case, each test can be designed based on the previous outcomes.  While tight information-theoretic limits and near-optimal practical algorithms are known for the adaptive setting in the absence of noise, surprisingly little is known in the noisy adaptive setting.  In this paper, we address this gap by providing information-theoretic achievability and converse bounds under various noise models, as well as a slightly weaker achievability bound for a computationally efficient variant.  These bounds are shown to be tight or near-tight in a broad range of scaling regimes, particularly at low noise levels.  The algorithms used for the achievability results have the notable feature of only using two or three stages of adaptivity.
\end{abstract}
\begin{IEEEkeywords}
    Group testing, sparsity, information-theoretic limits, adaptive algorithms
\end{IEEEkeywords}

\long\def\symbolfootnote[#1]#2{\begingroup\def\thefootnote{\fnsymbol{footnote}}\footnote[#1]{#2}\endgroup}

\symbolfootnote[0]{ The author is with the Department of Computer Science \& Department of Mathematics, National University of Singapore (e-mail: scarlett@comp.nus.edu.sg).  This work was supported by an NUS Startup Grant.}
\vspace*{-0.5cm}

%
%
\section{Introduction} \label{sec:intro}

The group testing problem consists of determining a small subset $S$ of ``defective'' items within a larger set of items $\{1,\dotsc,p\}$, based on a number of possibly-noisy tests. This problem has a history in medical testing \cite{Dor43}, and has regained significant attention following new applications in areas such as communication protocols \cite{Ant11}, pattern matching \cite{Cli10}, and database systems \cite{Cor05}, and connections with compressive sensing \cite{Gil08,Gil07}. In the noiseless setting, each test takes the form
\begin{equation}
    Y = \bigvee_{j \in S} X_j, \label{eq:gt_noiseless_model}
\end{equation}
where the test vector $X = (X_1,\dotsc,X_p) \in \{0,1\}^p$ indicates which items are included in the test, and $Y$ is the resulting observation.  That is, the output indicates whether at least one defective item was included in the test.   One wishes to design a sequence of tests $X^{(1)},\dotsc,X^{(n)}$, with $n$ ideally as small as possible, such that the outcomes can be used to reliably recover the defective set $S$ with probability close to one.

One of the defining features of the group testing problem is the distinction between the {\em non-adaptive} and {\em adaptive} settings.  In the non-adaptive setting, every test must be designed prior to observing any outcomes, whereas in the adaptive setting, a given test $X^{(i)}$ can be designed based on the previous outcomes $Y^{(1)},\dotsc,Y^{(i-1)}$.  It is of considerable interest to determine the extent to which this additional freedom helps in reducing the number of tests.

In the noiseless setting, a number of interesting results have been discovered along these lines:
\begin{itemize}
    \item When the number of defectives $k := |S|$ scales as $k = O(p^{1/3})$, the minimal number of tests permitting vanishing error probability scales as $n = \big(k\log_2\frac{p}{k}\big)(1+o(1))$ in both the adaptive and non-adaptive settings \cite{Bal13,Sca15b}.  Hence, at least information-theoretically, there is no asymptotic adaptivity gain.
    \item For scalings of the form $k = \Theta(p^{\theta})$ with $\theta \in \big(\frac{1}{3},1\big)$, the behavior $n = \big(k\log_2\frac{p}{k}\big)(1+o(1))$ remains unchanged in the adaptive setting \cite{Bal13}, but it remains open as to whether this can be attained non-adaptively.  For $\theta$ close to one, the best known non-adaptive achievability bounds are far from this threshold.
    \item Even in the first case above with no adaptivity gain, the adaptive algorithms known to achieve the optimal threshold are {\em practical}, having low storage and computation requirements \cite{Du93}.  In contrast, in the non-adaptive case, only computationally intractable algorithms have been shown to attain the optimal threshold \cite{Sca15b,Ald15}.
    \item It has recently been established that there is a provable adaptivity gap under certain scalings of the form $k = \Theta(p)$, i.e., the linear regime \cite{Abh18,Ald18}.
\end{itemize}
Despite this progress for the noiseless setting, there has been surprisingly little work on adaptivity in noisy settings; the vast majority of existing group testing algorithms for random noise models are non-adaptive \cite{Mal80,Cha11,Cha14,Sca17b}.  In this paper, we address this gap by providing new achievability and converse bounds for noisy adaptive group testing, focusing primarily on a widely-adopted symmetric noise model.  Before outlining our contributions, we formally introduce the setup.

\subsection{Problem Setup} \label{sec:setup}
 
Except where stated otherwise, we let the defective set $S$ be uniform on the ${p \choose k}$ subsets of $\{1,\dotsc,p\}$ of cardinality $k$.  As mentioned above, an adaptive algorithm iteratively designs a sequence of tests $X^{(1)},\dotsc,X^{(n)}$, with $X^{(i)} \in \{0,1\}^p$, and the corresponding outcomes are denoted by $\Yv = (Y^{(1)},\dotsc,Y^{(n)})$, with $Y^{(i)} \in \{0,1\}$.  A given test is allowed to depend on all of the previous outcomes.

Generalizing \eqref{eq:gt_noiseless_model}, we consider the following widely-adopted symmetric noise model:
 \begin{equation}
     Y = \bigvee_{j \in S} X_j \oplus Z, \label{eq:gt_symm_model}
 \end{equation}
 where $Z \sim \Bernoulli(\rho)$ for some $\rho \in \big(0,\frac{1}{2}\big)$, and $\oplus$ denotes modulo-2 addition.  In Section \ref{sec:asymm}, we will also consider other asymmetric noise models.
  
Given the tests and their outcomes, a \emph{decoder} forms an estimate $\Shat$ of $S$.  We consider the exact recovery criterion, in which the error probability is given by 
\begin{equation}
    \pe := \PP[\Shat \ne S], \label{eq:pe}
\end{equation}
and is taken over the randomness of the defective set $S$, the tests $X^{(1)},\dotsc,X^{(n)}$ (if randomized), and the noisy outcomes $Y^{(1)},\dotsc,Y^{(n)}$.
 
As a stepping stone towards exact recovery results, we will also consider a less stringent {\em partial recovery} criterion, in which we allow for up to $\dmax$ false positives and up to $\dmax$ false negatives, for some $\dmax > 0$.  That is, the error probability is
\begin{equation}
    \pe(\dmax) := \PP[d(S,\Shat) > \dmax], \label{eq:pe_dmax}
\end{equation}
where
\begin{equation}
    d(S,\Shat) = \max\{ |S \setminus \Shat|, |\Shat \setminus S| \}. \label{eq:dist}
\end{equation}
Understanding partial recovery is, of course, also of interest in its own right.  However, the results of \cite{Sca15b,Sca15} indicate that there is little or no adaptivity gain under this criterion, at least when $k = o(p)$ and $\dmax = \Theta(k)$.

Except where stated otherwise, we assume that the noise level $\rho$ and number of defectives $k$ are known.  In Section \ref{sec:conv}, we will consider cases where $k$ is only approximately known.
 
\subsection{Related work}

\noindent {\bf Non-adaptive setting.} The information-theoretic limits of group testing were first studied in the Russian literature \cite{Mal78,Mal80,Dya81}, and have recently become increasingly well-understood \cite{Ati12,Ald14,Sca15,Sca15b,Sca16b,Ald15}.  Among the existing works, the results most relevant to the present paper are as follows:
\begin{itemize}
    \item In the adaptive setting, it was shown by Baldassini {\em et al.} \cite{Bal13} that if the output $Y$ is produced by passing the noiseless outcome $U = \vee_{j \in S} X_j$ through a binary channel $P_{Y|U}$, then the number of tests for attaining $\pe \to 0$ must satisfy $n \ge \big(\frac{1}{C}k\log\frac{p}{k}\big)(1-o(1))$,\footnote{Here and subsequently, the function $\log(\cdot)$ has base $e$, and information measures have units of nats.} where $C$ is the Shannon capacity of $P_{Y|U}$ in nats.  For the symmetric noise model \eqref{eq:gt_symm_model}, this yields
    \begin{equation}
        n \ge \frac{k\log\frac{p}{k}}{\log 2 - H_2(\rho)} (1-o(1)), \label{eq:mi_conv}
    \end{equation}
    where $H_2(\rho) = \rho\log\frac{1}{\rho} + (1-\rho)\log\frac{1}{1-\rho}$ is the binary entropy function.
    \item In the non-adaptive setting with symmetric noise, it was shown that an information-theoretic threshold decoder attains the bound \eqref{eq:mi_conv} for $k = o(p)$ under the {partial recovery} criterion with $\dmax = \Theta(k)$ and an arbitrarily small implied constant \cite{Sca15b,Sca15}.  For exact recovery, a more complicated bound was also given in \cite{Sca15b} that matches \eqref{eq:mi_conv} when $k = \Theta(p^{\theta})$ for {\em sufficiently small} $\theta > 0$.
\end{itemize}
Several non-adaptive noisy group testing algorithms have been shown to come with rigorous guarantees.  We will use two of these non-adaptive algorithms as building blocks in our adaptive methods:
\begin{itemize}
    \item The {\em Noisy Combinatorial Orthogonal Matching Pursuit} (NCOMP) algorithm checks, for each item, the proportion of tests it was included in that returned positive, and declares the item to be defective if this number exceeds a suitably-chosen threshold.  This is known to provide optimal scaling laws for the regime $k = \Theta(p^{\theta})$ ($\theta \in (0,1)$) \cite{Cha11,Cha14}, albeit with somewhat suboptimal constants.
    \item The method of {\em separate decoding of items}, also known as {\em separate testing of inputs} \cite{Mal80,Sca17b}, also considers the items separately, but uses all of the tests.  Specifically, a given item's status is selected via a binary hypothesis test.  This method was studied for $k = O(1)$ in \cite{Mal80}, and for $k = \Theta(p^{\theta})$ in \cite{Sca17b}; in particular, it was shown that the number of tests is within a factor $\log 2$ of the optimal information-theoretic threshold under exact recovery as $\theta \to 0$, and under partial recovery (with $\dmax = \Theta(k)$) for all $\theta \in (0,1)$.
\end{itemize}
A different line of works has considered group testing with {\em adversarial} errors (e.g., see \cite{Mac97,Ngo11,Che13a}); these are less relevant to the present paper.

\noindent {\bf Adaptive setting.} As mentioned above, adaptive algorithms are well-understood in the noiseless setting \cite{Hwa72,Dam12}.  To our knowledge, the first algorithm that was proved to achieve $n = \big(k\log_2\frac{p}{k}\big)(1+o(1))$ for all $k = o(p)$ is Hwang's generalized binary splitting algorithm \cite{Hwa72,Du93}.  More recently, there has been interest in algorithms that only use limited rounds of adaptivity \cite{Mac98,Mez11,Dam12}, and among other things, it has been shown that the same guarantee can be attained using at most four stages \cite{Dam12}.  The two-stage setting is often considered to be of particular interest \cite{Deb05,Epp07,Mez11}.

In the noisy adaptive setting, the existing work is relatively limited.  In \cite{Cai13}, an adaptive algorithm called GROTESQUE was shown to provide optimal scaling laws in terms of {both samples and runtime}.  Our focus in this paper is only on the number of samples, but with a much greater emphasis on the {\em constant factors}.  In \cite[Ch.~4]{BalThesis}, noisy adaptive group testing algorithms were proposed for two different noise models based on the Z-channel and reverse Z-channel, also achieving an order-optimal required number of tests with reasonable constant factors.  We discuss these noise models further in Section \ref{sec:asymm}.

\subsection{Contributions}

In this paper, we characterize both the information-theoretic limits and the performance of practical algorithms for noisy adaptive group testing, characterizing the asymptotic required number of tests for $\pe \to 0$ as $p \to \infty$.  For the achievability part, we propose an adaptive algorithm whose first stage can be taken as any non-adaptive algorithm that comes with partial recovery guarantees, and whose second stage (and third stage in a refined version)  improve this initial estimate.  By letting the first stage use the information-theoretic threshold decoder of \cite{Sca15b}, we attain an achievability bound that is near-tight in many cases of interest, whereas by using separate decoding of items as per \cite{Mal80,Sca17b}, we attain a slightly weaker guarantee while still maintaining computational efficiency.  In addition, we provide a novel converse bound showing that $\Omega(k \log k)$ tests are always necessary, and hence, the implied constant in any scaling of the form $n = \Theta\big(k\log \frac{p}{k}\big)$ with $k = \Theta(p^{\theta})$ must grow unbounded as $\theta \to 1$.  (This is because $k \log k = \Theta\big( \frac{\theta}{1-\theta} k \log \frac{p}{k} \big)$ when $k = \Theta(p^{\theta})$.)

Our results are summarized in Figure \ref{fig:rates}, where we observe a considerable gain over the best known non-adaptive guarantees, particularly when the noise level $\rho$ is small.  Although there is a gap between the achievability and converse bounds for most values of $\theta$, the converse has the notable feature of showing that $n = \frac{k\log_2\frac{p}{k}}{\log 2 - H_2(\rho)} (1+o(1))$ is not always achievable, as one might conjecture based on the noiseless setting.  In addition, the gap between the (refined) achievability bound and the converse bound is zero or nearly zero in at least two cases: (i) $\theta$ is small; (ii) $\theta$ is close to one and $\rho$ is close to zero.   The algorithms used in our upper bounds have the notable feature of only using two or three rounds of adaptivity, i.e., two in the simple version, and three in the refined version.

In addition to these contributions for the symmetric noise model, we provide the following results for various other observation models:
\begin{itemize}
	\item In the noiseless case, we recover the threshold $n = \big(k \log_2\frac{p}{k}\big) (1+o(1))$ for all $\theta \in (0,1)$ using a two-stage adaptive algorithm.  Previously, the best known number of stages was four \cite{Dam12}.
	\item For the Z-channel noise model (defined formally in Section \ref{sec:asymm}), we show that one can attain $n = \frac{1}{C} \big(k \log_2\frac{p}{k}\big) (1+o(1))$ for all $\theta \in (0,1)$, where $C$ is the Shannon capacity of the channel.  This matches the general converse bound given in \cite{Bal13}, i.e., the generalized version of \eqref{eq:mi_conv}.  As a result, we improve on the above-mentioned bounds of \cite{BalThesis}, which contain reasonable yet strictly suboptimal constant factors.
	\item For the reverse Z-channel noise model (defined formally in Section \ref{sec:asymm}), we prove a similar converse bound to the one mentioned above for the symmetric noise model, thus showing that one cannot match the converse bound of \cite{Bal13} for all $\theta \in (0,1)$.  
\end{itemize}


The remainder of the paper is organized as follows.  For the symmetric noise model, we present the simple version of our achievability bound in Section \ref{sec:ach}, the refined version in Section \ref{sec:ach_ref}, and the converse in Section \ref{sec:conv}.  The other observation models mentioned above are considered in Section \ref{sec:asymm}, and conclusions are drawn in Section \ref{sec:conc}.
 
\begin{figure}
     \begin{centering}
         \includegraphics[width=0.48\columnwidth]{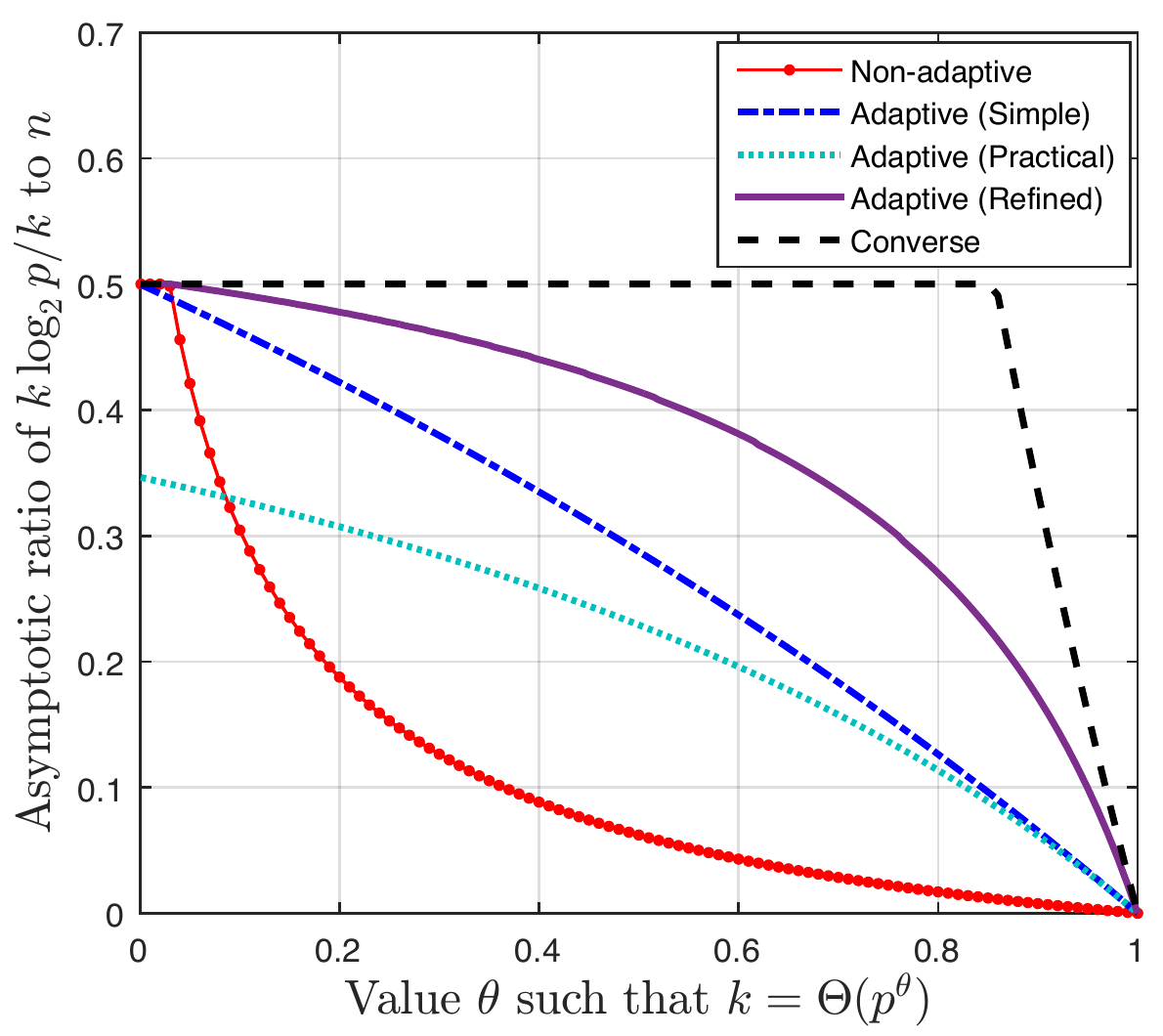} \quad
         \includegraphics[width=0.47\columnwidth]{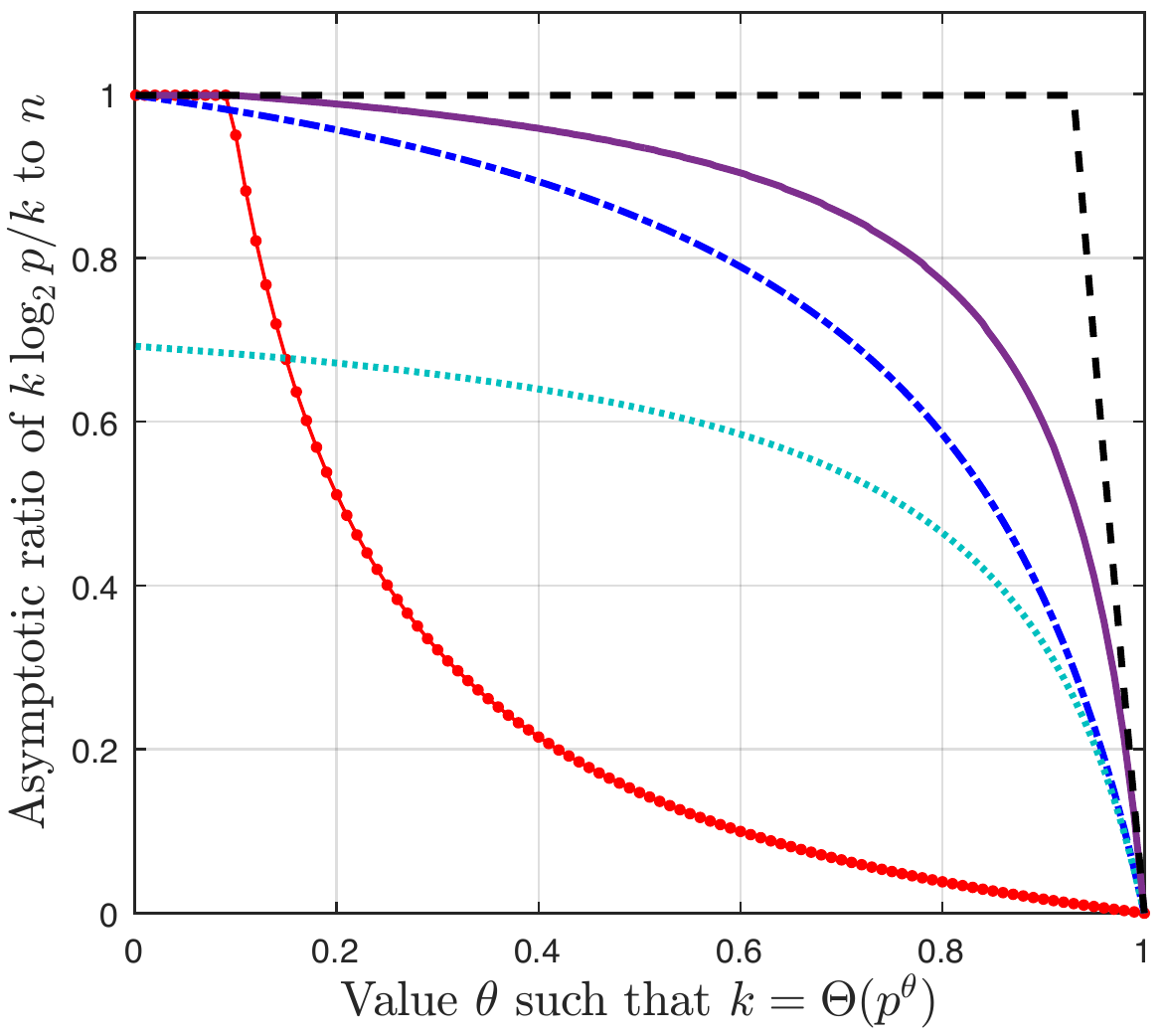}
         \par
     \end{centering}
     
     \caption{ Asymptotic bounds on the number of tests required for vanishing error probability under the noise levels $\rho = 0.11$ (Left) and $\rho = 10^{-4}$ (Right). \label{fig:rates} }
\end{figure}
 
\section{Achievability (Simple Version)} \label{sec:ach}

In this section, we formally state our simplest achievability results; a more complicated but also more powerful result is given in Section \ref{sec:ach_ref}.  Using a common two-stage approach, we provide achievability bounds for both a computationally intractable information-theoretic decoder, and computationally efficient decoder.

\subsection{Information-theoretic decoder}

The two-stage algorithm that we adopt is outlined {\em informally} in Algorithm \ref{alg:steps}; we describe the steps more precisely in the proof of Theorem \ref{thm:ach} below.    The idea is to use a non-adaptive algorithm with {partial recovery} guarantees, and then refine the solution by resolving the false negatives and false positives separately, i.e., Steps 2a and 2b.  While these latter steps are stated separately in Algorithm \ref{alg:steps}, the tests that they use can be performed together in a single round of adaptivity, so that the overall algorithm is a two-stage procedure.

\begin{algorithm}[H]
	\caption*{ \manuallabel{pr:procedure}{1} \textbf{Algorithm 1:} Two-stage algorithm for noisy group testing (informal). \label{alg:steps} }
	
	\begin{enumerate}
		\item Apply the information-theoretic threshold decoder of \cite{Sca15b} (see Appendix \ref{app:partial_gamma}) to the ground set $\{1,\dotsc,p\}$ to find an estimate $\Shat_1$ of cardinality $k$ such that 
		\begin{equation}
			\max\{|\Shat_1 \backslash S|, |S \backslash \Shat_1|\} \le \alpha_1 k
		\end{equation}
		with high probability, for some small $\alpha_1 > 0$.
		\item[2a.] Apply a variation of NCOMP \cite{Cha11} (see Appendix \ref{app:ncomp}) to the reduced ground set $\{1,\dotsc,p\} \backslash \Shat_1$ to exactly identify the false negatives $S \setminus \Shat_1$ from the first step.  Let these items be denoted by $\Shat'_{2a}$.
		\item[2b.] Test the items in $\Shat_1$ individually $\ntil$ times (for some $\ntil$ to be specified), and let $\Shat'_{2b}$ contain the items that returned positive at least $\frac{\ntil}{2}$ times.  The final estimate of the defective set is given by $\Shat := \Shat'_{2a} \cup \Shat'_{2b}$.
	\end{enumerate}
\end{algorithm}

\noindent Our first main information-theoretic achievability result is as follows.

\begin{thm} \label{thm:ach}
    Under the symmetric noisy group testing setup with crossover probability $\rho \in \big(0,\frac{1}{2}\big)$, with $k = \Theta(p^{\theta})$ for some $\theta \in (0,1)$, there exists a two-stage adaptive group testing algorithm such that
    \begin{equation}
        n \le \bigg( \frac{k\log\frac{p}{k}}{\log 2 - H_2(\rho)} + \frac{k\log k}{\frac{1}{2} \log \frac{1}{4\rho(1-\rho)}} \bigg) (1+o(1)) \label{eq:ach}
    \end{equation}
    and such that $\pe \to 0$ as $p \to \infty$.
\end{thm}
\begin{proof}
    We study the guarantees of the three steps in Algorithm \ref{alg:steps}, and the number of tests used for each one.
    
    {\bf Step 1.} It was shown in \cite{Sca15b} that, for an arbitrarily small constant $\alpha_1 > 0$, there exists a non-adaptive group testing algorithm returning some set $\Shat_1$ of cardinality $k$ such that 
    \begin{equation}
        \max\{|\Shat_1 \backslash S|, |S \backslash \Shat_1|\} \le \alpha_1 k, \label{eq:s1_success}
    \end{equation}
    with probability approaching one, with the number of tests being at most
    \begin{equation}
        n_1 \le \bigg( \frac{k\log\frac{p}{k}}{\log 2 - H_2(\rho)} \bigg) (1+o(1)).
    \end{equation}
    In Appendix \ref{app:partial_gamma}, we recap the decoding algorithm and its analysis.  The non-adaptive test design for this stage is the ubiquitous i.i.d.~Bernoulli design.
    
    {\bf Step 2a.} Let us condition on the first step being successful, in the sense that \eqref{eq:s1_success} holds.  We claim that there exists a non-adaptive algorithm that, when applied to the reduced ground set $\{1,\dotsc,p\} \backslash \Shat_1$, returns $\Shat'_{2a}$ containing precisely the set of (at most $\alpha_1 k$) defective items $S \setminus \Shat_1$ with probability approaching one, with the number of samples behaving as
    \begin{equation}
        n_2 = O\bigg( \alpha_1 k \log \frac{p}{\alpha_1 k} \bigg).
    \end{equation}
    If the number of defectives $k_1 := |S \backslash \Shat_1|$ in the reduced ground set were known, this would simply be an application of the $O(k_1 \log p)$ scaling derived in \cite{Cha11} for the NCOMP algorithm.  In Appendix \ref{app:ncomp}, we adapt the algorithm and analysis of \cite{Cha11} to handle the case that $k_1$ is only known up to a constant factor.
    
    In fact, in the present setting, we only know that $k_1 \in [0,\alpha_1 k]$, so we do not even know $k_1$ up to a constant factor.  To get around this, we apply a simple trick that is done purely for the purpose of the analysis:  Instead of applying the modified NCOMP algorithm directly to $\{1,\dotsc,p\} \backslash \Shat_1$, apply it to the slightly larger set in which $\alpha_1 k$ ``dummy'' defective items are included.  Then, the number of defectives is in $[\alpha_1 k,2\alpha_1 k]$, and is known up to a factor of two.  We do not expect that this trick would ever be useful practice, but it is convenient for the sake of the analysis.
    
    {\bf Step 2b.} Since we conditioned on the first step being successful, at most $\alpha_1 k$ of the $k$ items in $\Shat_1$ are non-defective.  In the final step, we simply test each item in $\Shat_1$ individually $\ntil$ times, and declare the item positive if and only if at least half of the outcomes are positive.
    
    To study the success probability, we use a well-known Chernoff-based concentration bound for Binomial random variables: If $Z \sim \mathrm{Binomial}(N,q)$, then
    \begin{align}
        &\PP[ Z \le nq' ] \le e^{-N D_2(q' \| q)}, \quad q' < q, \label{eq:bino_conc1}
    \end{align}
    where $D_2(q' \| q) = q'\log\frac{q'}{q} + (1-q')\log\frac{1-q'}{1-q}$ is the binary KL divergence function.
    
    Fix an arbitrary item $j$, and let $\Ntil_{j,1}$ be the number of its $\ntil$ tests that are positive.  Since the test outcomes are distributed as $\mathrm{Bernoulli}(1-\rho)$ for defective $j$ and $\mathrm{Bernoulli}(\rho)$ for non-defective $j$, we obtain from \eqref{eq:bino_conc1} that
    \begin{align}
        &\PP\bigg[ \Ntil_{j,1} \le \frac{\ntil}{2} \bigg] \le e^{-\ntil D_2(\frac{1}{2} \| 1-\rho)} = e^{-\ntil D_2(\frac{1}{2} \| \rho)}, \quad j \in S \label{eq:indiv_test_1} \\
        &\PP\bigg[ \Ntil_{j,1} \ge \frac{\ntil}{2} \bigg] \le e^{-\ntil D_2(\frac{1}{2} \| \rho)}, \phantom{= e^{-\ntil D_2(\frac{1}{2} \| 1-\rho)}} \hspace*{2.85ex} j \notin S. \label{eq:indiv_test_2} 
    \end{align}
    Hence, we obtain from the union bound over the $k$ items in $\Shat_1$ that
    \begin{equation}
       \PP[\Shat'_{2b} \ne (S \cap \Shat_1)] \le k \cdot e^{-\ntil D_2(\frac{1}{2} \| \rho)}.
    \end{equation}
    For any $\eta > 0$, the right-hand side tends to zero as $p \to \infty$ under the choice
    \begin{equation}
        \ntil = \frac{\log k}{D_2(\frac{1}{2} \| \rho)}(1+\eta),
    \end{equation}
    which gives a total number of tests in Step 2b of
    \begin{equation}
        n_{2b} = \frac{k\log k}{D_2(\frac{1}{2} \| \rho)}(1+\eta).
    \end{equation}
    The proof is concluded by noting that $\eta$ can be arbitrarily small, and writing $D_2(\frac{1}{2} \| \rho) = \frac{1}{2}\log \frac{1}{4\rho(1-\rho)}$.
\end{proof}

A weakness of Theorem \ref{thm:ach} is that it does not achieve the threshold $n = \frac{k \log\frac{p}{k}}{\log 2 - H_2(\rho)} (1+o(1))$ for any value of $\theta > 0$ (see Figure \ref{fig:rates}), even though such a threshold is achievable for sufficiently small $\theta$ even non-adaptively \cite{Sca15b}.  We overcome this limitation via a refined three-stage algorithm in Section \ref{sec:ach_ref}.

\subsection{Practical decoder}

Of the three steps given in Algorithm \ref{alg:steps} and the proof of Theorem \ref{thm:ach}, the only one that is computationally demanding is the first, which uses an information-theoretic threshold decoder to identify $S$ up to a distance ({\em cf.}, \eqref{eq:dist}) of $d(S,\Shat) \le \alpha_1 k$, for small $\alpha_1 > 0$.  A similar approximate recovery result was also shown in \cite{Sca17b} for separate decoding of items, which is computationally efficient.  The asymptotic threshold on $n$ for separate decoding of items is only a $\log 2$ factor worse than the optimal information-theoretic threshold \cite{Sca17b}, and this fact leads to the following counterpart to Theorem \ref{thm:ach}.

\begin{thm} \label{thm:ach_sdi}
    Under the symmetric noisy group testing setup with crossover probability $\rho \in \big(0,\frac{1}{2}\big)$, and $k = \Theta(p^{\theta})$ for some $\theta \in (0,1)$, there exists a computationally efficient two-stage adaptive group testing algorithm such that
    \begin{equation}
    n \le \bigg( \frac{k\log\frac{p}{k}}{\log 2 \cdot  (\log 2 - H_2(\rho))} + \frac{k\log k}{\frac{1}{2} \log \frac{1}{4\rho(1-\rho)}} \bigg) (1+o(1))
    \end{equation}
    and $\pe \to 0$ as $p \to \infty$.
\end{thm}

\noindent The proof is nearly identical to that of Theorem \ref{thm:ach}, except that the required number of tests in the first stage is multiplied by $\frac{1}{\log 2}$ in accordance with \cite{Sca17b}.  For brevity, we omit the details.

\section{Achievability (Refined Version)} \label{sec:ach_ref}

As mentioned previously, a weakness of Theorem \ref{thm:ach} is that it only achieves the behavior $n \le \big( \frac{k \log\frac{p}{k}}{\log 2 - H_2(\rho)} \big)(1+o(1))$ (for which a matching converse is known \cite{Bal13}) in the limit as $\theta \to 0$, even though this can be achieved even non-adaptively for sufficiently small $\theta$ \cite{Sca15b}.  Since adaptivity provides extra freedom in the design, we should expect the corresponding bounds to be at least as good as the non-adaptive setting.

While we can simply take the better of Theorem \ref{thm:ach} and the exact recovery result of \cite{Sca15b}, this is a rather unsatisfying solution, and it leads to a discontinuity in the asymptotic threshold ({\em cf.}, Figure \ref{fig:rates}). It is clearly more desirable to construct an adaptive scheme that ``smoothly'' transitions between the two.  In this section, we attain such an improvement by modifying Algorithm \ref{alg:steps} in two ways.  The resulting algorithm is outlined informally in Algorithm \ref{alg:steps_ref}, and the modifications are as follows:
\begin{itemize}
    \item In the first stage, instead of learning $S$ up to a distance of $\alpha_1 k$ for some constant $\alpha_1 \in (0,1)$, we  learn it up to a distance of $k^{\gamma}$ for some $\gamma \in (0,1)$.  The non-adaptive partial recovery analysis of \cite{Sca15b} requires non-trivial modifications for this purpose; we provide the details in Appendix \ref{app:partial_gamma}.
    \item We split Step 2b of Algorithm \ref{alg:steps} into two stages, one comprising Step 2b in Algorithm \ref{alg:steps_ref}, and the other comprising Step 3.  The former of these identifies most of the defective items, and the latter resolves the rest.
\end{itemize}
It is worth noting that, at least using our analysis techniques, neither of the above modifications alone is enough to obtain a bound that is always at least as good as the non-adaptive exact recovery result of \cite{Sca15b}.  We will shortly see, however, that the two modifications combined do suffice.

\begin{algorithm}[H]
    \caption*{ \manuallabel{pr:procedure}{2} \textbf{Algorithm 2:} Three-stage algorithm for noisy group testing (informal). \label{alg:steps_ref} }
    
    \begin{enumerate}
        \item[1.] Apply the information-theoretic threshold decoder of \cite{Sca15b} (see Appendix \ref{app:partial_gamma}) to the ground set $\{1,\dotsc,p\}$ to find an estimate $\Shat_1$ of cardinality $k$ such that 
        \begin{equation}
        \max\{|\Shat_1 \backslash S|, |S \backslash \Shat_1|\} \le k^{\gamma}
        \end{equation}
        with high probability, where $\gamma \in (0,1)$.
        \item[2a.] Apply a variation of NCOMP \cite{Cha11} (see Appendix \ref{app:ncomp}) to the reduced ground set $\{1,\dotsc,p\} \backslash \Shat_1$ to exactly identify the false negatives from the first step.  Let these items be denoted by $\Shat'_{2a}$.
        \item[2b.]  Test each item in $\Shat_1$ individually $\ncheck$ times (for some $\ncheck$ to be specified), and let $\Shat'_{2b} \subseteq \Shat_1$ contain the $k - \alpha_2 k$ items that returned positive the highest number of times, for some small $\alpha_2> 0$.
        \item[3.]  Test the items in $\Shat_1 \setminus \Shat'_{2b}$ (of which there are $\alpha_2 k$) individually $\ntil$ times (for some $\ntil$ to be specified), and let $\Shat'_3$ contain the items that returned positive at least $\frac{\ntil}{2}$ times.  The final estimate of the defective set is given by $\Shat := \Shat'_{2a} \cup \Shat'_{2b} \cup \Shat'_3$.
    \end{enumerate}
\end{algorithm}

%

The following theorem characterizes the asymptotic number of tests required.

\begin{thm} \label{thm:ach2}
    Under the symmetric noisy group testing setup with crossover probability $\rho \in \big(0,\frac{1}{2}\big)$, under the scaling $k = \Theta(p^{\theta})$ for some $\theta \in (0,1)$, there exists a three-stage adaptive group testing algorithm such that
    \begin{equation}
    n \le \inf_{\gamma \in (0,1), \delta_2 \in (0,1)} \Big(\max\big\{ \nMIi, \nMIii(\gamma,\delta_2), \nConc(\gamma,\delta_2) \big\} + \nIndiv(\gamma)\Big) (1+o(1)) \label{eq:ach2}
    \end{equation}
    and $\pe \to 0$ as $p \to \infty$, where:
    \begin{itemize}
        \item The standard mutual information based term is
        \begin{equation}
        \nMIi = \frac{ k\log\frac{p}{k} }{ \log 2 - H_2(\rho) }. \label{eq:nMI1}
        \end{equation}
        \item An additional mutual information based term is
        \begin{equation}
        \nMIii(\gamma,\delta_2) = \frac{2}{(\log 2)(1-2\rho)\log\frac{1-\rho}{\rho}} \cdot \frac{1}{1-\delta_2} \cdot \Big( (1-\theta) k\log p + 2(1-\gamma)k\log k \Big). \label{eq:nMI2}
        \end{equation}
        \item The term associated with a concentration bound is
        \begin{equation}
        \nConc(\gamma,\delta_2) = \frac{4(1+\frac{1}{3}\delta_2(1-2\rho)) }{ (\log 2) \delta_2^2 (1-2\rho)^2 } \cdot (1-\gamma) k \log k. \label{eq:nConc}
        \end{equation}
        \item The term associated with individual testing is
        \begin{equation}
            \nIndiv(\gamma) = \frac{\gamma k \log k}{ D_2(\rho \| 1 - \rho) }. \label{eq:nIndiv}
        \end{equation}
    \end{itemize}
\end{thm}

\noindent While the theorem statement is somewhat complex, it is closely related to other simpler results on group testing:
\begin{itemize}
    \item In the limit as $\gamma \to 0$, the term $\max\big\{ \nMIi, \nMIii(\delta_2), \nConc(\gamma,\delta_2) \big\}$ corresponds to the condition for exact recovery derived in \cite{Sca15b}.   Since $\nIndiv(\gamma)$ becomes negligible as $\gamma \to 0$, this means that we have the above-mentioned desired property of being at least as good as the exact recovery result.
    \item Taking $\gamma \to 1$ and $\delta_2 \to 0$ in a manner such that $\frac{1-\gamma}{\delta_2^2} \to 0$, we recover a strengthened version of Theorem \ref{thm:ach} with $D\big( \frac{1}{2} \| 1-\rho) = \frac{1}{2} \log \frac{1}{4\rho(1-\rho)}$ increased to $D\big( \rho \| 1-\rho)$.\footnote{By letting the first stage of Algorithm \ref{alg:steps_ref} use separate decoding of items \cite{Sca17b}, one can obtain a strengthened version of Theorem \ref{thm:ach_sdi} with the same improvement.  This result is omitted for the sake of brevity, as the main purpose of the refinements given in this section is to obtain a bound that is always at least as good as the non-adaptive information-theoretic bound of \cite{Sca15b}.}
\end{itemize}
The parameter $\delta_2$ controls the trade-off between the concentration behavior associated with $\nConc$ and the mutual information based term $\nMIii$.  


\subsection{Proof of Theorem \ref{thm:ach2}}

The proof follows similar steps to those of Theorem \ref{thm:ach}, considering the four steps of Algorithm \ref{alg:steps_ref} separately. 

{\bf Step 1.} We show in Appendix \ref{app:partial_gamma} that the approximate recovery result of \cite{Sca15b} can be extended as follows: There exists a non-adaptive algorithm recovering an estimate $\Shat_1$ of cardinality $k$ such that $d(S,\Shat_1) \le k^{\gamma}$ with probability approaching one, provided that the number of tests $n_1$ satisfies
\begin{equation}
    n_1 \ge \max\big\{ \nMIi, \nMIii(\gamma,\delta_2), \nConc(\gamma,\delta_2) \big\} \cdot (1+o(1))
\end{equation}
for some $\delta_2 \in (0,1)$, under the definitions in \eqref{eq:nMI1}--\eqref{eq:nConc}.  This algorithm and its corresponding estimate $\Shat_1$ constitute the first step.

{\bf Step 2a.} The algorithm and analysis for Stage 2 are identical to that of Theorem \ref{thm:ach}: We use the variation of NCOMP given in Appendix \ref{app:ncomp} to identify all defective items in $\{1,\dotsc,p\} \setminus \Shat_1$ with probability approaching one, while only using $O(k^{\gamma} \log p) = o(k \log p)$ tests.

{\bf Step 2b.} For this step, we need to show that the set $\Shat'_{2b}$ constructed in Algorithm \ref{alg:steps_ref} only contains defective items.   Recall that this set is constructed by testing each item in $\Shat_1$ individually $\ncheck$ times, and keeping items that returned positive the highest number of times.  Since $\Shat'_{2b}$ contains $|\Shat_1| - \alpha_2 k$ items, requiring all of these items to be defective is equivalent to requiring that the set of $\alpha_2 k$ items with the smallest number of positive outcomes includes the $k^{\gamma}$ (or fewer) non-defective items in $\Shat_1$.  For any $\zeta > 0$, the following two conditions suffice for this purpose:
\begin{itemize}
    \item Event $\Ac_1$:  All non-defective items in $\Shat_1$ return positive less than $\zeta \ncheck$ times;
    \item Event $\Ac_2$: At most $\alpha_2 k - k^{\gamma}$ defective items return positive less than $\zeta \ncheck$ times.
\end{itemize}
Here we assume that $k$ is sufficiently large so that $\alpha_2 k > k^{\gamma}$, which is valid since $\gamma < 1$ and $\alpha_2 > 0$ are constant.

Fix an arbitrary item $j$, and let $\Ncheck_{j,1}$ be the number of its $\ncheck$ tests that are positive.  Since the test outcomes are distributed as $\mathrm{Bernoulli}(1-\rho)$ for defective $j$ and $\mathrm{Bernoulli}(\rho)$ for non-defective $j$, we obtain from \eqref{eq:bino_conc1} that
\begin{align}
&\PP\big[ \Ncheck_{j,1} \le \zeta\ncheck \big] \le e^{-\ncheck D_2(\zeta \| 1-\rho)}, \quad j \in S \label{eq:indiv_test_1a} \\
&\PP\big[ \Ncheck_{j,1} \ge \zeta\ncheck \big] \le e^{-\ncheck D_2(\zeta \| \rho)}, \quad~~\, j \notin S. \label{eq:indiv_test_2a} 
\end{align}
Hence, we obtain from the union bound over the non-defective items in $\Shat_1$ that
\begin{equation}
\PP[\Ac_1^c] \le k^{\gamma} \cdot e^{-\ncheck D_2(\zeta \| \rho)},
\end{equation}
which is upper bounded by $\delta_3 > 0$ as long as 
\begin{equation}
    \ncheck \ge \frac{  \log \frac{k^\gamma}{ \delta_3 } }{  D_2(\zeta \| \rho) }. \label{eq:ncheck_cond}
\end{equation}
Moreover, regarding the event $\Ac_2$, the average number of defective items that return positive less than $\zeta\ncheck$ times is upper bounded by $k e^{-\ncheck D_2(\zeta \| 1-\rho)}$ (recall that $|\Shat_1| = k$), and hence, Markov's inequality gives
\begin{equation}
    \PP[\Ac_2^c] \le \frac{ k e^{-\ncheck D_2(\zeta \| 1-\rho)} }{ \alpha_2k - k^{\gamma} }.
\end{equation}
This is upper bounded by $\frac{ \frac{k}{\log k} }{ \alpha_2 k - k^{\gamma} } \to 0$ as long as $\ncheck \ge \frac{\log \log k}{ D_2(\zeta \| 1-\rho) }$.  This, in turn, behaves as $o(\log k)$ for any $\zeta < 1- \rho$.  Hence, we are left with only the condition on $\ncheck$ in \eqref{eq:ncheck_cond}, and choosing $\zeta$ arbitrarily close to $1- \rho$ means that we only need the following to hold for arbitrarily small $\eta > 0$:
\begin{equation}
    \ncheck \ge \frac{  \gamma \log k }{  D_2(1-\rho \| \rho) } (1+\eta),
\end{equation}
since $\log \frac{k^\gamma}{ \delta_3 } = (\gamma \log k)(1+o(1))$ for arbitrarily small $\delta_3$.  Multiplying by $k$ (i.e., the number of items that are tested individually $\ncheck$ times) and noting that $D(1-\rho \| \rho) = D(\rho \| 1- \rho)$, we deduce that the number of tests in this stage is asymptotically at most $\nIndiv(\gamma)$, defined in \eqref{eq:nIndiv}.

{\bf Step 3.} This step is the same as that of Step 2b in Algorithm \ref{alg:steps}, but we are now working with $\alpha_2 k$ items rather than $k$ items.  As a result, the number of tests required is $O( \alpha_2 k\log k)$, meaning that the coefficient to $k \log k$ can be made arbitrarily small by a suitable choice of $\alpha_2$.





\section{Converse} \label{sec:conv}

To our knowledge, the best-known existing converse bound for the symmetric noise model in the adaptive setting is the capacity-based bound of \cite{Bal13}, shown in \eqref{eq:mi_conv}.  On the other hand, the preceding achievability bounds contains $k \log k $ terms, meaning that the gap between the achievability and converse grows unbounded as $\theta \to 1$ under the scaling $k = \Theta(p^{\theta})$ (since $k \log k = \Theta\big( \frac{\theta}{1-\theta} k \log \frac{p}{k} \big)$).   In this section, we provide a novel converse bound revealing that $\Omega(k \log k)$ behavior is unavoidable.

There is a minor caveat to this converse result:  We have not been able to prove it in the case that $S$ is known to have cardinality exactly $k$, but rather, only in the case that it is known to have cardinality either $k$ or $k-1$.  We strongly conjecture that this distinction has no impact on the fundamental limits; we argue in Appendix \ref{app:partial_prob} that Theorem \ref{thm:ach} remains true even when $k$ is only known up to a multiplicative $1+o(1)$ term, and Theorem \ref{thm:ach2} remains true when $k$ is only known up to an additive $o(k^{\gamma})$ term.  Since we assume that $k \to \infty$, these assumptions are much milder than the assumption $|S| \in \{k-1,k\}$

To make the model definition more precise, fix $k \le 2p$, and define
\begin{equation}
    \Sc_{k,p} = \big\{ S\subseteq\{1,\dotsc,p\} \,:\, |S|=k \big\},
\end{equation}
and similarly for $\Sc_{k-1,p}$. We consider the following distribution for the random defective set:
\begin{equation}
    S \sim \mathrm{Uniform}( \Sc_{k,p} \cup \Sc_{k-1,p} ). \label{eq:S_modified}
\end{equation}
Under this slightly modified model, we have the following.

\begin{thm} \label{thm:conv}
    Consider the symmetric noisy group testing setup with crossover probability $\rho \in \big(0,\frac{1}{2}\big)$, $S$ distributed according to \eqref{eq:S_modified}, and $k \to \infty$ with $k \le \frac{p}{2}$.  For any adaptive algorithm, in order to achieve $\pe \to 0$, it is necessary that
    \begin{equation}
    n \ge \max\bigg\{ \frac{ k\log\frac{p}{k} }{ \log 2 - H_2(\rho) }, \frac{k \log k}{ \log\frac{1-\rho}{\rho} } \bigg\} (1-o(1)). \label{eq:conv}
    \end{equation}
\end{thm}
\begin{proof}
    See Section \ref{sec:conv_proof}.
\end{proof}

The first term is precisely \eqref{eq:mi_conv}, so our novelty is in deriving the second term.  This result provides the first counter-example to the natural conjecture that the optimal number of tests is $\big(\frac{ k\log\frac{p}{k} }{ \log 2 - H_2(\rho) }\big)(1+o(1))$ whenever $k = \Theta(p^{\theta})$ with $\theta \in (0,1)$.  Indeed, the $\Omega(k \log k)$ lower bound reveals that the constant pre-factor to $k\log\frac{p}{k}$ must grow unbounded as $\theta \to 1$.

It is interesting to observe the behavior of Theorems \ref{thm:ach}, \ref{thm:ach2}, and \ref{thm:conv} in the limit as $\rho \to 0$.  As one should expect, under the scaling $k = \Theta(p^{\theta})$ for fixed $\theta \in (0,1)$, both the achievability and converse bounds (see \eqref{eq:ach} and \eqref{eq:conv}) tend towards the noiseless limit $\big(k\log_2\frac{p}{k}\big)(1+o(1))$ as $\rho \to 0$.  Moreover, the achievability and converse bounds scale similarly with respect to $\rho$, in the sense that the $k \log k$ term is scaled by $\Theta\big( \frac{1}{\log\frac{1}{\rho}} \big)$ in both cases.   

In fact, if we consider the refined achievability bound (Theorem \ref{thm:ach2}), we can make a stronger claim.  If we take $\gamma \to 1$ and $\theta \to 1$ simultaneously, then the bound in \eqref{eq:ach2} is asymptotically equivalent to $\nIndiv(1)$, since $\nMIi$ scales as $k \log \frac{p}{k} \ll k \log k$, whereas the constant factors in $\nMIii$ and $\nConc$ vanish (see \eqref{eq:nMI1}--\eqref{eq:nConc}).  Hence, we are only left with $\nIndiv(1)$ in \eqref{eq:nIndiv}, and if $\rho$ is small, then the denominator $D_2(\rho \| 1-\rho) = \rho \log \frac{\rho}{1-\rho} + (1-\rho)\log\frac{1-\rho}{\rho}$ is approximately equal to $\log\frac{1}{\rho}$.  The exact same statement is true for the denominator in \eqref{eq:conv}, and hence, the achievability and converse bounds exhibit {\em matching constant factors}.  Specifically, this statement holds when the order of the limits is first $n \to \infty$, then $\theta \to 1$, then $\rho \to 0$.  This fact explains the near-identical behavior of the achievability and converse in Figure \ref{fig:rates} for $\theta$ close to one in the low noise setting, $\rho = 10^{-4}$.

On the other hand, for fixed $\theta \in (0,1)$, the logarithmic decay of the $\Theta\big(\frac{1}{\log \frac{1}{\rho}}\big)$ factor to zero is quite slow, which explains the non-negligible deviation from the noiseless threshold (i.e., a straight line at height $1$) in Figure \ref{fig:rates}, even in the low-noise case.

Another interesting consequence of Theorem \ref{thm:conv} is that in the linear regime $k = \Theta(p)$, one requires $n = \Omega(p \log p)$ in the presence of noise.  This is in stark contrast to the noiseless setting, where individual testing trivially identifies $S$ with only $p$ tests.   

The proof of Theorem \ref{thm:conv} is inspired by that of a converse bound for the top-$m$ arm identification problem from the multi-armed bandit (MAB) literature \cite{Kal12}.  Compared with the latter, the adaptive group testing setting has a number of distinct features that are non-trivial to handle:
\begin{itemize}
    \item In group testing, one does not necessarily test one item at a time, whereas in the MAB setting of \cite{Kal12}, one pulls one arm at a time.
    \item In contrast with \cite{Kal12}, we do not consider a minimax lower bound, but rather, a Bayesian lower bound for a given distribution on $S$.  The latter is more difficult, in the sense that a Bayesian lower bound implies a minimax lower bound but not vice versa.
    \item In our setting, the status of each item is binary-valued (defective or non-defective), whereas the construction of a hard MAB problem in \cite{Kal12} consists of three distinct types of items (or ``arms'' in the MAB terminology), corresponding to high reward, medium reward, and low reward.
\end{itemize}
We now proceed with the proof.  

\subsection{Proof of Theorem \ref{thm:conv}} \label{sec:conv_proof}

We assume without loss of generality that any given test $X^{(i)}$ is deterministic given $Y^{(1)},\dotsc,Y^{(i-1)}$, and that the final estimate $\Shat$ is similarly deterministic given the test outcomes.  To see that it suffices to consider this case, we note that
\begin{equation}
    \PP[\mathrm{error}] = \EE\big[ \PP[\mathrm{error} \,|\, \Ac] \big] \ge \min_{A} \PP[\mathrm{error} \,|\, \Ac = A],
\end{equation}
where $\Ac$ denotes a randomized algorithm (i.e., combination of test design and decoder), and $A$ is a realization of $\Ac$ corresponding to a deterministic algorithm.

Suppose that after $S$ is randomly generated according to \eqref{eq:S_modified}, a genie reveals $S \cup T$ to the decoder, where $T$ is a uniformly random set of non-defective items such that $|S \cup T| = 2k$ (i.e., $T$ has cardinality $2k - |S| \in \{k,k+1\}$).  Hence, we are left with an easier group testing problem consisting of $2k$ items, $k-1$ or $k$ of which are defective.  Since the prior distribution on $S$ in \eqref{eq:S_modified} is uniform, we see that conditioned on the ground set of size $2k$, the defective set $S$ is uniform on the ${2k \choose k} + {2k \choose k-1}$ possibilities.  

Without loss of generality, assume that the $2k$ revealed items are $\{1,\dotsc,2k\}$, and hence, the new distribution of $S$ given the information from the genie is
\begin{equation}
S \sim \mathrm{Uniform}( \Sc_{k,2k} \cup \Sc_{k-1,2k} ).
\end{equation}
We first study the error probability conditioned on a given defective set $S \subset \{1,\dotsc,2k\}$ having cardinality $k$.  For any such fixed choice, we denote probabilities and expectations (with respect to the noisy outcomes) by $\PP_S$ and $\EE_S$.

Fix $\epsilon \in \big(0,\frac{1}{2}\big)$, and for each $j \in S$, let $N_j$ be the (random) number of tests containing item $j$ and no other defective items.  Since $\sum_{j \in S} N_j \le n$ with probability one, we have $\sum_{j \in S} \EE_S[N_j] \le n$, meaning that at most $(1-\epsilon)k$ of the $j \in S$ have $\EE_S[N_j] \ge \frac{n}{(1-\epsilon)k}$.  For all other $j$, we have $\EE_S[N_j] \le \frac{n}{(1-\epsilon)k}$, and Markov's inequality gives $\PP_S[N_j \ge \frac{(1+2\epsilon)n}{k}] \le \frac{1}{(1-\epsilon)(1+2\epsilon)} < 1$.  Denoting $\psi(\epsilon) := \frac{1}{(1-\epsilon)(1+2\epsilon)}$ for brevity, we have proved the following.

\begin{lem} \label{lem:js}
    For any $\epsilon \in \big(0,\frac{1}{2}\big)$, and any set $S \subset \{1,\dotsc,2k\}$ of cardinality $k$, there exist at least $\epsilon k$ items $j \in S$ such that $\PP_S[N_j \ge \frac{(1+2\epsilon)n}{k}] \le \psi(\epsilon)$, where $\psi(\epsilon) = \frac{1}{(1-\epsilon)(1+2\epsilon)}$.
\end{lem}

The following lemma, consisting of a change of measure between the probabilities under two different defective sets, will also be crucial.  Recalling that we are considering test designs that are deterministic given the past samples, we see that $N_j$ is a deterministic function of $\Yv = (Y^{(1)},\dotsc,Y^{(n)})$, so we write the corresponding function as $n_j(\yv)$.  Moreover, we let  $\Yc_S$ be the set of $\yv$ sequences that are decoded as $S$, and we write $\PP[\yv]$ and $\PP[\Yc_S]$ as shorthands for $\PP[\Yv = \yv]$ and $\PP[\Yv \in \Yc_S]$, respectively.

\begin{lem} \label{lem:chg_msr}
    Given $S$ of cardinality $k$, for any $j \in S$, and any output sequence $\yv$ such that $n_j(\yv) \le \frac{(1+2\epsilon )n}{k}$, we have
    \begin{equation}
    \PP_{S \setminus \{j\}}[\yv] \ge \PP_S[\yv] \Big( \frac{\rho}{1-\rho} \Big)^{\frac{(1+2\epsilon)n}{k}}. \label{eq:chg_msr}
    \end{equation}
    Moreover, if $j \in S$ is such that $\PP_S[N_j \ge \frac{(1+2\epsilon)n}{k}] \le \psi(\epsilon)$, then
    \begin{equation}
    \PP_{S \setminus \{j\}}[\Yc_S] \ge \big(\PP_S[\Yc_S] - \psi(\epsilon) \big) \Big( \frac{\rho}{1-\rho} \Big)^{\frac{(1+2\epsilon)n}{k}}. \label{eq:chg_msr2}
    \end{equation}
\end{lem}
\begin{proof}
    Again using the fact that the test designs that are deterministic given the past samples, we can write
    \begin{align}
        \PP_S[\yv] 
            &= \prod_{i=1}^n \PP_S[y^{(i)} | y^{(1)},\dotsc,y^{(i-1)}] \\
            &= \prod_{i=1}^n \PP_S[y^{(i)} | x^{(i)}],  \label{eq:p_product}
    \end{align}
    where $x^{(i)} \in \{0,1\}^p$ is the $i$-th test.  Note that \eqref{eq:p_product} holds because $Y^{(i)}$ depends on the previous samples only through $X^{(i)}$.  An analogous expression also holds for $\PP_{S \backslash \{j\}}[\yv]$.
    
    Due to the ``or'' operation in the observation model \eqref{eq:gt_symm_model}, the only tests for which the outcome probability changes as a result of removing $j$ from $S$ are those for which $j$ was the unique defective item tested.  We have at most $\frac{(1+2\epsilon)n}{k}$ such tests by assumption, and each of them causes the probability of $y^{(i)}$ (given $x^{(i)}$) to be multiplied or divided by $\frac{\rho}{1-\rho}$.  Since $\rho < 0.5$, we deduce the lower bound in \eqref{eq:chg_msr}, corresponding to the case that all $\frac{(1+2\epsilon)n}{k}$ of them are multiplied by this factor.
    
    To prove the second part, we write
    \begin{align}
    \PP_{S \setminus \{j\}}[\Yc_S]
    &\ge \PP_{S \setminus \{j\}}\bigg[\Yv \in \Yc_S \,\cap\, N_j \le \frac{(1+2\epsilon)n}{k} \bigg] \\
    &\ge \PP_{S}\bigg[\Yv \in \Yc_S \,\cap\, N_j \le \frac{(1+2\epsilon)n}{k} \bigg] \Big( \frac{\rho}{1-\rho} \Big)^{\frac{(1+2\epsilon)n}{k}} \label{eq:err_bound_3} \\
    &\ge \big(\PP_S[\Yc_S] - \psi(\epsilon)\big) \Big( \frac{\rho}{1-\rho} \Big)^{\frac{(1+2\epsilon)n}{k}}, \label{eq:err_bound_4}
    \end{align}
    where \eqref{eq:err_bound_3} follows from the first part of the lemma, and \eqref{eq:err_bound_4} follows by writing $\PP[A \cap B] \ge \PP[A] - \PP[B^c]$.
\end{proof}

The idea behind applying this lemma is that if a given $\yv$ is decoded to $S$, then it cannot be decoded to $S \setminus \{j\}$; hence, if a given sequence $\yv$ contributes to $\PP_S[\mathrm{no~error}]$, then it also contributes to $\PP_{S \setminus \{j\}}[\mathrm{error}]$.   We formalize this idea as follows.  Recalling that $\Sc_{k,2k}$ is the set of all subsets of $\{1,\dotsc,2k\}$ of cardinality $k$, we have
\begin{align}
\sum_{S' \in \Sc_{k-1,2k}} \PP_{S'}[\mathrm{error}]
&\ge \sum_{S' \in \Sc_{k-1,2k}} \sum_{j \notin S'} \PP_{S'}[\Yc_{S' \cup \{j\}}] \label{eq:avg_chg_msr_1} \\
&= \sum_{S' \in \Sc_{k-1,2k}} \sum_{j \notin S'} \sum_{S \in \Sc_{k,2k}} \openone\big\{ S = S' \cup \{j\} \big\} \PP_{S'}[\Yc_{S}] \label{eq:avg_chg_msr_2} \\
&= \sum_{S' \in \Sc_{k-1,2k}} \sum_{j=1}^{2k} \sum_{S \in \Sc_{k,2k}} \openone\big\{ S = S' \cup \{j\} \big\} \PP_{S'}[\Yc_{S}] \label{eq:avg_chg_msr_3} \\
&= \sum_{S \in \Sc_{k,2k}} \sum_{j \in S} \sum_{S' \in \Sc_{k-1,2k}}  \openone\big\{ S = S' \cup \{j\} \big\} \PP_{S'}[\Yc_{S}] \label{eq:avg_chg_msr_4} \\        
&= \sum_{S \in \Sc_{k,2k}} \sum_{j \in S} \PP_{S \backslash \{j\}}[\Yc_{S}], \label{eq:avg_chg_msr_5}
\end{align}
where \eqref{eq:avg_chg_msr_1} follows since $S'$ differs from $S' \cup \{j\}$, \eqref{eq:avg_chg_msr_2} follows since the indicator function is only equal to one for $S = S' \cup \{j\}$, \eqref{eq:avg_chg_msr_3} follows since the extra $j$ included in the middle summation (i.e., $j \in S$) also make the indicator function equal zero, \eqref{eq:avg_chg_msr_4} follows by re-ordering the summations and noting that the indicator function equals zero when $j \notin S$, and \eqref{eq:avg_chg_msr_5} follows by only keeping the $S'$ for which the indicator function is one.

The following lemma is based on lower bounding \eqref{eq:avg_chg_msr_5} using Lemma \ref{lem:chg_msr}.

\begin{lem} \label{lem:cnv_final}
    If $\frac{1}{|\Sc_{k,2k}|} \sum_{S \in \Sc_{k,2k}} \PP_{S}[\mathrm{error}] \le \delta$ for some $\delta > 0$, then
    \begin{equation}
        \frac{1}{|\Sc_{k-1,2k}|} \sum_{S' \in \Sc_{k-1,2k}} \PP_{S'}[\mathrm{error}] \ge \epsilon k \cdot \big( 1-2\delta - \psi(\epsilon) \big) \cdot\Big( \frac{\rho}{1-\rho} \Big)^{\frac{(1+2\epsilon)n}{k}} \label{eq:cnv_final}
    \end{equation}
    for any $\epsilon \in \big(0,\frac{1}{2}\big)$.
\end{lem}
\begin{proof}
    Since $\frac{1}{|\Sc_{k,2k}|} \sum_{S \in \Sc_{k,2k}} \PP_{S}[\mathrm{error}] \le \delta$ and $|\Sc_{k,2k}| = {2k \choose k}$, there must exist at least $\frac{1}{2}{2k \choose k}$ defective sets $S \in \Sc_{k,2k}$ such that $\PP_{S}[\mathrm{error}] \le 2\delta$.  We lower bound the first summation in \eqref{eq:avg_chg_msr_5} by a summation over such $S$,
    and for each one, we lower bound the summation over $j \in S$ by the set of size at least $\epsilon k$ given in Lemma \ref{lem:js}.  For the choices of $S$ and $j$ that are kept in this lower bound, the summand $\PP_{S \setminus \{j\}}[\Yc_S]$ is lower bounded by $\big( 1 - 2\delta - \psi(\epsilon) \big)\big( \frac{\rho}{1-\rho} \big)^{\frac{(1+2\epsilon)n}{k}}$ by the second part of Lemma \ref{lem:chg_msr} (with $\PP_S[\Yc_S] = \PP_S[\mathrm{no~error}] \ge 1-2\delta$).  Putting this all together, we obtain
    \begin{equation}
    \sum_{S' \in \Sc_{k-1,2k}} \PP_{S'}[\mathrm{error}] \ge \frac{1}{2}{2k \choose k} \cdot \epsilon k \cdot \big( 1-2\delta - \psi(\epsilon) \big) \cdot\Big( \frac{\rho}{1-\rho} \Big)^{\frac{(1+2\epsilon)n}{k}}.
    \end{equation}
    Using the identity ${2k \choose k} = {2k \choose k-1} \cdot \frac{2k - k}{k} = 2{2k \choose k-1}$, this yields
    \begin{equation}
    \frac{1}{{2k \choose k-1}}\sum_{S' \in \Sc_{k-1,2k}} \PP_{S'}[\mathrm{error}] \ge \epsilon k \cdot \big( 1-2\delta - \psi(\epsilon) \big) \cdot\Big( \frac{\rho}{1-\rho} \Big)^{\frac{(1+2\epsilon)n}{k}},
    \end{equation}
    which proves the lemma.
\end{proof}

Recalling that $\psi(\epsilon) = \frac{1}{(1-\epsilon)(1+2\epsilon)}$, it is easily verified that $\psi(\epsilon) < 1$ for all $\epsilon \in \big(0,\frac{1}{2}\big)$.  Hence, by a suitable choice of $\delta$, we can let $\epsilon$ be arbitrarily small while still ensuring that $1-2\delta - \psi(\epsilon) > 0$.  Moreover, $\PP[S \in \Sc_{k,2k}]$ and $\PP[S \in \Sc_{k-1,2k}]$ are both bounded away from zero under the distribution in \eqref{eq:S_modified}.  Most importantly, the term $k\big( \frac{\rho}{1-\rho} \big)^{\frac{(1+2\epsilon)n}{k}}$ appearing in \eqref{eq:cnv_final} is lower bounded by $\delta' > 0$ as long as $n \le \frac{k \log (k\delta')}{ (1+2\epsilon)\log\frac{1-\rho}{\rho} }$.  Since $\epsilon$ may be arbitrarily small and $\log(k\delta') = (\log k)(1+o(1))$, we deduce that the following condition is necessary for attaining arbitrarily small error probability:
\begin{equation}
    n \ge \frac{k\log k}{\log\frac{1-\rho}{\rho}} (1-\eta),
\end{equation}
where $\eta > 0$ is arbitrarily small.  This completes the proof of Theorem \ref{thm:conv}.

\section{Other Observation Models} \label{sec:asymm}

While we have focused on the symmetric noise model \eqref{eq:gt_symm_model} for concreteness, most of our algorithms and analysis techniques can be extended to other observation models.  In this section, we present some of the resulting bounds for three different models: The noiseless model \eqref{eq:gt_noiseless_model}, the Z-channel model,
\begin{gather}
    P_{Y|U}(0|0) = 1, \quad P_{Y|U}(1|0) = 0, \\
    P_{Y|U}(0|1) = \rho, \quad P_{Y|U}(1|1) = 1-\rho,
\end{gather}
and the reverse Z-channel model, 
\begin{gather}
    P_{Y|U}(0|0) = 1-\rho, \quad P_{Y|U}(1|0) = \rho, \\
    P_{Y|U}(0|1) = 0, \quad P_{Y|U}(1|1) = 1,
\end{gather}
where in both cases we define $U = \vee_{j \in S} X_j$.  That is, we pass the noiseless observation through the suitable binary channel; see Figure \ref{fig:Zchannels} for an illustration.  Under the Z-channel model, positive tests indicate with certainty that a defective item is included, whereas under the reverse Z-channel model, negative tests indicate with certainty that no defective item is included.  While the two channels have the same capacity, it is interesting to ask whether one of the two is fundamentally more difficult to handle in the context of group testing.  We provide a partial answer to this question in the adaptive setting; see also \cite{Sca18b} for the non-adaptive setting.
    
\begin{figure}
    \begin{centering}
        \includegraphics[width=0.2\columnwidth]{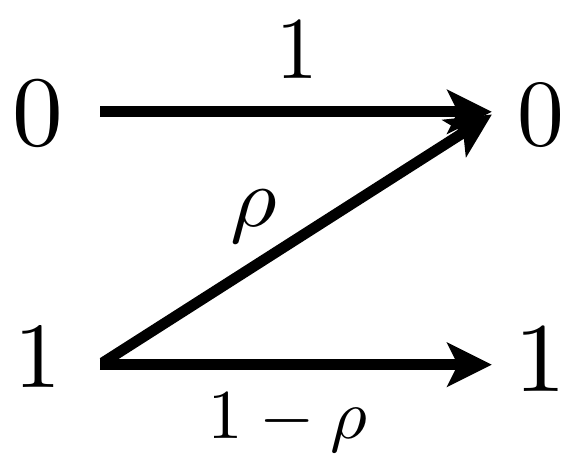} \qquad
        \includegraphics[width=0.2\columnwidth]{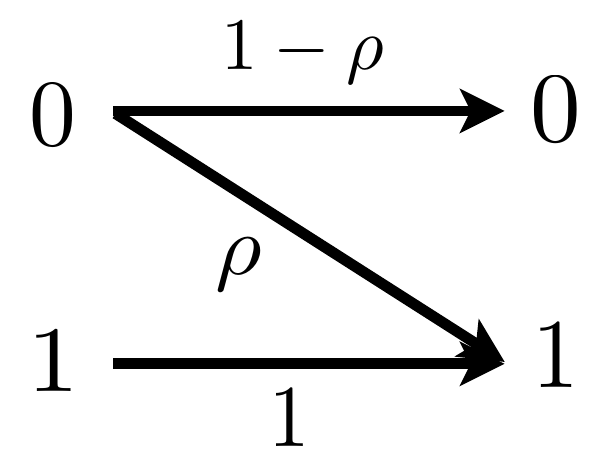} 
        \par
    \end{centering}
    
    \caption{Z-channel (Left) and reverse Z-channel (Right). \label{fig:Zchannels}}
\end{figure}

\subsection{Noiseless setting}

In the noiseless setting, the final step of Algorithm \ref{alg:steps} is much simpler: Simply test the items in $\Shat_2$ individually once each.  This only requires $k$ tests, and succeeds with certainty, yielding the following.

\begin{thm} \label{thm:noiseless}
    Under the scaling $k = \Theta(p^{\theta})$ for some $\theta \in (0,1)$, there exists a two-stage algorithm for noiseless adaptive group testing that succeeds with probability approaching one, with a number of tests bounded by
    \begin{equation}
    n \le \bigg(k \log_2\frac{p}{k}\bigg)(1+o(1)). \label{eq:noiseless_result}
    \end{equation}
    Moreover, there exists a computationally efficient two-stage algorithm that succeeds with probability approaching one, with a number of tests bounded by
    \begin{equation}
    n \le \frac{1}{\log 2}\bigg(k \log_2\frac{p}{k}\bigg)(1+o(1)). \label{eq:noiseless_result2}
    \end{equation}
\end{thm}

The upper bound \eqref{eq:noiseless_result} is tight, as it matches the so-called counting bound, e.g., see \cite{Joh15}.  To our knowledge, the minimum number of stages used to attain this bound previously for all $\theta \in (0,1)$ was four \cite{Dam12}.  It is worth noting, however, that the algorithm of \cite{Dam12} has low computational complexity, unlike Algorithm \ref{alg:steps}.


The bound \eqref{eq:noiseless_result} does not contradict the converse bound of M\'ezard and Toninelli \cite{Mez11}; the latter states that any two-stage algorithm with {\em zero} error probability must have an {\em average} number of tests of $\frac{1}{\log 2} \big(k \log_2\frac{p}{k}\big)(1+o(1))$ or higher.  In contrast, \eqref{eq:noiseless_result} corresponds to {\em vanishing} error probability and a {\em fixed} number of tests.

\subsection{Z-channel model}

Under the Z-channel model, the capacity-based converse bound of \cite{Bal13} turns out to be tight for all $\theta \in (0,1)$, as stated in the following.

\begin{thm} \label{thm:ach_Z}
    Under the noisy group testing model with Z-channel noise having parameter $\rho \in (0,1)$, and a number of defectives satisfying $k = \Theta(p^{\theta})$ for some $\theta \in (0,1)$, there exists a three-stage adaptive algorithm achieving vanishing error probability with
    \begin{equation}
        n \le \frac{k\log\frac{p}{k}}{ C(\rho) } (1+o(1)), \label{eq:n_Z}
    \end{equation}
    where $C(\rho)$ is the capacity of the Z-channel in nats.
\end{thm}
\begin{proof}
    The analysis is similar to that of the symmetric noise model ({\em cf.}, Theorem \ref{thm:ach2}), so we omit most of the details.  
    
    In the first stage, we use i.i.d.~Bernoulli testing with parameter $\nu > 0$ chosen to ensure that the induced distribution $P_U$ of $U = \vee_{j \in S} X_i$ equals the capacity-achieving input distribution of the Z-channel.  Under this choice, a straightforward extension of the analysis of \cite{Sca15b} (see the final part of Appendix \ref{app:partial_gamma} for details) reveals that we can find a set $\Shat_1$ of cardinality $k$ such that $d(S,\Shat_1) \le \alpha_1 k$ with $n$ satisfying \eqref{eq:n_Z}, where $d$ is defined in \eqref{eq:dist}, and $\alpha_1 > 0$ is arbitrarily small.
    
    The second stage is similar to steps 2a and 2b in Algorithm \ref{alg:steps_ref}.  The modifications required in step 2a are stated in Appendix \ref{app:ncomp}, and step 2b is in fact simpler:  We include a given item in $\Shat'_{2b}$ if and only if {\em any} of its tests returned positive.  Due to the nature of the Z-channel, no non-defectives are included in $\Shat'_{2b}$.  On the other hand, the probability of a positive item returning negative on all $\ncheck$ tests is given by $\rho^{\ncheck}$, and is asymptotically vanishing if $\ncheck = \log \log k$ (say).  Hence, by Markov's inequality, we have with probability approaching one that the number of defective items that fail to be placed in $\Shat'_{2b}$ is smaller than $\alpha_1 k$ with probability approaching one.  Moreover, the required number of tests is $O(k \log \log k)$, which is asymptotically negligible.
    
    In the third stage, as in Algorithm \ref{alg:steps_ref}, we test each item individually $\ntil$ times.  Here, however, we let $\Shat'_3$ contain the items that returned positive in {\em any} test.  There are again no false positives, and a given defective item is a false negative with probability $\rho^{\ntil}$.  By the union bound and the fact that there are at most $2\alpha_1 k$ items and $\alpha_1 k$ defective items, we readily deduce vanishing error probability as long as $\ntil = O(\log k)$, meaning the total number of tests is $O(\alpha_1 k \log k)$.  This is asymptotically negligible, since $\alpha_1$ is arbitrarily small.
\end{proof}

This result shows that under Z-channel noise, the conjecture of the optimal (inverse) coefficient to $k\log\frac{p}{k}$ equaling the channel capacity (e.g., see \cite{Bal13}) is true for all $\theta \in (0,1)$, in stark contrast to the symmetric noise model.

It is worth noting that the converse analysis of Section \ref{sec:conv} does not apply to the Z channel model.  This is because any analog of Lemma \ref{lem:chg_msr} is impossible:  If there exists a test outcome $y_i = 1$ where $j$ is the only defective included, then $\PP_{S \setminus \{j\}}[\yv] = 0$, meaning we cannot hope for an inequality of the form \eqref{eq:chg_msr}.

\subsection{Reverse Z-channel model}

Under the reverse Z-channel model, we have the following analog of the converse bound in Theorem \ref{thm:conv}.

\begin{thm} \label{thm:conv_RZ}
    Consider the noisy group testing setup with reverse Z-channel noise having parameter $\rho \in (0,1)$, $S$ distributed according to \eqref{eq:S_modified}, and $k \to \infty$ with $k \le \frac{p}{2}$.  For any adaptive algorithm, in order to achieve $\pe \to 0$, it is necessary that
    \begin{equation}
    n \ge \max\bigg\{ \frac{ k\log\frac{p}{k} }{ C(\rho) }, \frac{k \log k}{ \log\frac{1}{\rho} } \bigg\} (1-o(1)), \label{eq:conv_RZ}
    \end{equation}
    where $C(\rho)$ is the capacity of the Z-channel in nats.
\end{thm}
\begin{proof}
    The first bound in \eqref{eq:conv_RZ} is the capacity-based bound from \cite{Bal13}.  On the other hand, the second bound follows from a near-identical analysis to the proof of Theorem \ref{thm:conv}, with the only difference being that $\frac{\rho}{1-\rho}$ is replaced by $\rho$ in \eqref{eq:chg_msr} and the subsequent equations that make use of \eqref{eq:chg_msr}.
    
    We note that unlike the $Z$-channel, the cases where one of $\PP_S[\yv]$ and $\PP_{S \setminus \{j\}}[\yv]$ is zero and the other is non-zero are not problematic.  Specifically, this only occurs when $\PP_S[\yv] = 0$, and in this case, any inequality of the form \eqref{eq:chg_msr} is trivially true.
\end{proof}

Interestingly, this result shows that reverse Z-channel noise is more difficult to handle than Z-channel noise by an arbitrarily large factor as $\theta$ gets closer to one, even though the two channels have the same capacity. 

\section{Conclusion} \label{sec:conc}

\noindent We have developed both information-theoretic limits and practical performance guarantees for noisy adaptive group testing.  Some of the main implications of our results include the following:
\begin{itemize}
    \item Under the scaling $k = \Theta(p^{\theta})$, for most $\theta \in (0,1)$, our information-theoretic achievability guarantees for the symmetric noise model are significantly better than the best known non-adaptive achievability guarantees, and similarly when it comes to practical guarantees.
    \item Our converse for the symmetric noise model reveals that $n = \Omega(k \log k)$ is necessary, and hence, the implied constant to $n = \Theta\big(k \log \frac{p}{k} \big)$ must grow unbounded as $\theta \to 1$.  This phenomenon also holds true for the reverse Z-channel noise model, but not for the Z-channel noise model.
    \item Our bounds are tight or near-tight in several cases of interest, including small values of $\theta$, and low noise levels with $\theta$ close to one.  Moreover, in the noiseless case, we obtain the optimal threshold using a two-stage algorithm; previously the smallest known number of stages was four.
\end{itemize}
It is worth noting that our two-stage (or three-stage) algorithm and its analysis remain applicable when {\em any} non-adaptive algorithm is used in the first stage, as long as it identifies a suitably high fraction of the defective set.  Hence, improved practical or information-theoretic guarantees for partial recovery in the non-adaptive setting immediately transfer to improved exact recovery guarantees in the adaptive setting.  

\appendix

\subsection{Non-Adaptive Partial Recovery Result with $\dmax = \Theta(k^{\gamma})$} \label{app:partial_gamma}

The analysis of \cite{Sca15b} considers the case that the maximum distance ({\em cf.}, \eqref{eq:pe_dmax}--\eqref{eq:dist}) scales as $\dmax = \Theta(k)$.  In this section, we adapt the analysis therein to the case $\dmax = \Theta(k^{\gamma})$ for some $\gamma \in (0,1)$.  This generalization is useful for the refined achievability bound given in Section \ref{sec:ach_ref} ({\em cf.}, Theorem \ref{thm:ach2}), and is also of interest in its own right.

\subsubsection{Notation} Recall that $S$ is uniform on the set of subsets of $\{1,\dotsc,p\}$ having a given cardinality $k$.  As in \cite{Sca15b}, we consider non-adaptive i.i.d.~Bernoulli testing, where each item is placed in a given test with probability $\frac{\nu}{k}$ for some $\nu > 0$.  We focus our attention on $\nu = \log 2$, though we will still write $\nu$ for the parts of the analysis that apply more generally.  The test matrix is denoted by $\Xv \in \{0,1\}^{n \times p}$ (i.e., the $i$-th row is $X^{(i)}$), and the notation $\Xv_s$ denotes the sub-matrix obtained by keeping only the columns indexed by $s \subseteq \{1,\dotsc,p\}$.

Next, we recall some notation from \cite{Sca15b}.  It will prove convenient to work with random variables that are implicitly conditioned on a fixed value of $S$, say $s=\{1,\dotsc,k\}$.  We write $P_{Y|X_{s}}$ for the conditional test outcome probability, where $X_s$ is the subset of the test vector $X$ indexed by $s$.  Moreover, we write
\begin{align}
    P_{X_s Y}(x_s,y) &:= P_{X}^{k}(x_s)P_{Y|X_s}(y|x_s), \label{eq:distr_sl} \\
P_{\Xv_s\Yv}(\xv_s,\yv) &:= P_X^{n \times k}(\xv_s) P^{n}_{Y|X_s}(\yv|\xv_{s}), \label{eq:vector_distr} 
\end{align}
where $P^{n}_{Y|X_s}(\cdot|\cdot)$ is the $n$-fold product of $P_{Y|X_s}(\cdot|\cdot)$, and $P_X^{(\cdot)}$ denotes the i.i.d.~$\mathrm{Bernoulli}\big(\frac{\nu}{k}\big)$ distribution for a vector or matrix of the size indexed in the superscript.  The random variables $(X_s,Y)$ and $(\Xv_s,\Yv)$ are distributed as
\begin{align}
    (X_s,Y) &\sim  P_{X_s Y}, \label{eq:joint_dist} \\
(\Xv_s,\Yv) &\sim P_{\Xv_s\Yv}, \label{eq:joint_dist_n}
\end{align}
and the remaining entries of the measurement matrix are distributed as $\Xv_{s^c} \sim P_X^{n \times (p-k)}$, independent of $(\Xv_s,\Yv)$.

In our analysis, we consider partitions of the defective set $s$ into two sets $\sdif \ne \emptyset$ and $\seq$.  One can think of $\seq$ as corresponding to an overlap $s \cap \sbar$ between the true set $s$ and some incorrect set $\sbar$, with $\sdif$ corresponding to the indices $s \backslash \sbar$ in one set but not the other.  For a fixed defective set $s$, and a corresponding pair $(\sdif,\seq)$, we write
\begin{align}
    P_{Y|X_{\sdif}X_{\seq}}(y|x_{\sdif},x_{\seq}) &:= P_{Y|X_s}(y|x_s), \label{eq:p_split2}
\end{align}
where $P_{Y|X_s}$ is the marginal distribution of \eqref{eq:vector_distr}.  This form of the conditional error probability allows us to introduce the marginal distribution
\begin{align}
    P_{Y|X_{\seq}}(y|x_{\seq}) &:= \sum_{x_{\sdif}} P_X^{\ell}(x_{\sdif})P_{Y|X_{\sdif}X_{\seq}}(y|x_{\sdif},x_{\seq}),
\end{align}
where $\ell := |\sdif| = k - |\seq|$.  Using the preceding definitions, we introduce the \emph{information density} \cite{Pol10}
\begin{align}
    \imath^n(\xv_{\sdif}; \yv | \xv_{\seq} ) &:= \sum_{i=1}^{n} \imath(x^{(i)}_{\sdif}; y^{(i)} | x^{(i)}_{\seq}), \label{eq:idens_bn} \\
    \imath(x_{\sdif}; y | x_{\seq}) &:= \log\frac{ P_{Y|X_{\sdif}X_{\seq}}(y|x_{\sdif},x_{\seq}) }{ P_{Y|X_{\seq}}(y|x_{\seq}) } \label{eq:idens_b}
\end{align}
where $(\cdot)^{(i)}$ denotes the $i$-th entry (respectively, row) of a vector (respectively, matrix).  Averaging \eqref{eq:idens_b} with respect to $(X_s,Y)$ in \eqref{eq:joint_dist} yields a conditional mutual information, which we denote by
\begin{equation}
    I_{\ell} := I(X_{\sdif};Y|X_{\seq}), \label{eq:condMI}
\end{equation}
where $\ell := |\sdif|$; by symmetry, the mutual information for each $(\sdif,\seq)$ depends only on this quantity.

\subsubsection{Choice of decoder} We use the same information-theoretic threshold decoder as that in \cite{Sca15b}: Fix the constants $\{\gamma_{\ell}\}_{\ell = \dmax + 1}^{k}$, and search for a set $s$ of cardinality $k$ such that
\begin{equation}
    \imath^n(\Xv_{\sdif};\Yv|\Xv_{\seq}) \ge \gamma_{|\sdif|}, \quad \forall (\sdif,\seq)\text{ such that }|\sdif| > \dmax. \label{eq:dec}
\end{equation}
If multiple such $s$ exist, or if none exist, then an error is declared.  This decoder is inspired by analogous thresholding techniques from the channel coding literature \cite{Fei54,Han03}.

\subsubsection{Useful existing results}

We build heavily on several intermediate results given in \cite{Sca15}, stated as follows:
\begin{itemize}
	\item {\bf Initial bounds.} Since the analysis is the same for any defective set $s$ of cardinality $k$, we assume without loss of generality that $s = \{1,\dotsc,k\}$.  The initial non-asymptotic bound of \cite{Sca15b} takes the form
	\begin{equation}
        \pe(\dmax) \le \PP\bigg[ \bigcup_{(\sdif,\seq)\,:\,|\sdif| > \dmax} \bigg\{ \imath^n(\Xv_{\sdif}; \Yv | \Xv_{\seq} ) \le \log {{p-k} \choose |\sdif|} + \log\bigg(\frac{k}{\delta_1}{k \choose |\sdif|}\bigg) \bigg\} \bigg] + \delta_1
	\end{equation}
	for any $\delta_1  > 0$. A simple consequence of this non-asymptotic bound is the following: For any positive constants $\{\delta_{2,\ell}\}_{\ell = \dmax+1}^k$, if the number of tests is at least
	\begin{equation}
       n \ge \max_{\ell=\dmax+1,\dotsc,k} \frac{ \log{{p-k} \choose \ell} + \log\big(\frac{k}{\delta_1}{k \choose \ell}\big) }{ I_{\ell} (1-\delta_{2,\ell}) }, \label{eq:final_ach}
	\end{equation}
    and if each information density satisfies a concentration bound of the form
    \begin{equation}
         \PP\big[ \imath^n(\Xv_{\sdif}; \Yv | \Xv_{\seq}) \le n(1 - \delta_{2,\ell})I_{\ell} \big] \le \psi_{\ell}(n,\delta_{2,\ell}), \label{eq:psi_ach}
    \end{equation}
    for some functions $\{\psi_{\ell}\}_{\ell=\dmax+1}^{k}$, then
    \begin{equation}
        \pe(\dmax) \le \sum_{\ell=\dmax + 1}^k{k \choose \ell}\psi_\ell(n,\delta_{2,\ell}) + \delta_1. \label{eq:nonasymp_ach}
    \end{equation}
	\item {\bf Characterization of mutual information.} Under the symmetric noise model with crossover probability $\rho \in \big(0,\frac{1}{2}\big)$, the conditional mutual information $I_{\ell}$ behaves as follows as $k \to \infty$:
    \begin{itemize}
        \item If $\frac{\ell}{k} \to 0$, then
        \begin{equation}
            I_{\ell} = \bigg(e^{-\nu}\nu \frac{\ell}{k} (1-2\rho)\log\frac{1-\rho}{\rho}\bigg)(1+o(1)). \label{eq:gtn_I_ord}
        \end{equation}
        \item If $\frac{\ell}{k} \to \alpha \in (0,1]$, then
        \begin{equation}
            I_{\ell} = e^{-(1-\alpha)\nu} \big(H_2\big(e^{-\alpha \nu} \star \rho\big) - H_2(\rho)\big) (1+o(1)). \label{eq:gtn_I_const}
        \end{equation}
    \end{itemize}
	\item {\bf Concentration bounds.} The following concentration bounds provide explicit choices for $\psi_{\ell}$ satisfying \eqref{eq:psi_ach}:
    \begin{itemize}
        \item For all $\ell$ and $\delta > 0$, we have
        \begin{equation}
            \PP\Big[ \big|\imath^n(\Xv_{\sdif}; \Yv | \Xv_{\seq}) - nI_{\ell} \big| \ge n\delta \Big] \le 2\exp\bigg(- \frac{\delta^2 n}{4(8+\delta)} \bigg) \label{eq:conc_gen_disc}
        \end{equation}
        for all $(\sdif,\seq)$ with $|\sdif|=\ell$.
        \item If $\frac{\ell}{k} \to 0$, then for any $\epsilon > 0$ and $\delta_{2} > 0$ (not depending on $p$), the following holds for sufficiently large $p$:
        \begin{equation}
            \PP\Big[ \imath^n(\Xv_{\sdif}; \Yv | \Xv_{\seq}) \le nI_{\ell}(1-\delta_{2}) \Big] \le \exp\bigg(-n\frac{\ell}{k} e^{-\nu}\nu\bigg( \frac{\delta_2^2 (1-2\rho)^2}{2(1+\frac{1}{3}\delta_2(1-2\rho))} \bigg)(1-\epsilon)\bigg). \label{eq:conc_gt_noisy}
        \end{equation}
        for all $(\sdif,\seq)$ with $|\sdif|=\ell$.
    \end{itemize}
     
\end{itemize}

\noindent With these tools in place, we proceed by obtaining an explicit bound on the number of tests for the case $\dmax = \Theta(k^{\gamma})$.  

\subsubsection{Bounding the error probability} We split the summation over $\ell$ in \eqref{eq:nonasymp_ach} into two terms:
\begin{equation}
    T_1 := \sum_{\ell=\dmax + 1}^{\frac{k}{\sqrt{\log k}}}{k \choose \ell}\psi_\ell(n,\delta_{2}^{(1)}), \quad T_2 := \sum_{\ell=\frac{k}{\sqrt{\log k}}}^{k}{k \choose \ell}\psi_\ell(n,\delta_{2}^{(2)}), \label{eq:T1T2}
\end{equation}
where we have let $\delta_{2,\ell}$ equal a given value $\delta_{2}^{(1)} \in (0,1)$ for all $\ell$ in the first sum, and a different value $\delta_{2}^{(2)} \in (0,1)$ for all $\ell$ in the second sum.

To bound $T_1$, we consider $\psi_\ell(n,\delta_{2})$ equaling the right-hand side of \eqref{eq:conc_gt_noisy}.  Letting $c(\delta_2) = e^{-\nu}\nu\big( \frac{\delta_2^2 (1-2\rho)^2}{2(1+\frac{1}{3}\delta_2(1-2\rho))} \big)(1-\epsilon)$ for brevity, we have
\begin{equation}
    T_1 \le k \max_{\ell=\dmax + 1, \dotsc, \frac{k}{\sqrt{\log k}}} {k \choose \ell} e^{-n\cdot\frac{\ell}{k}\cdot c(\delta_2^{(1)})},
\end{equation}
where we have upper bounded the summation defining $T_1$ by $k$ times the maximum.  Re-arranging, we find in order to attain $T_1 \le \delta_1$, it suffices that
\begin{equation}
    n \ge \max_{\ell=\dmax + 1, \dotsc, \frac{k}{\sqrt{\log k}}}  \frac{1}{c(\delta_2^{(1)})} \cdot \frac{k}{\ell} 
    \cdot \bigg( \log {k \choose \ell} + \log\frac{k}{\delta_1} \bigg).
\end{equation}
Writing $\log {k \choose \ell} = \big(\ell \log\frac{k}{\ell}\big)(1+o(1))$, this simplifies to 
\begin{align}
    n &\ge \max_{\ell=\dmax + 1, \dotsc, \frac{k}{\sqrt{\log k}}}  \frac{1}{c(\delta_2^{(1)})} \cdot \bigg( k\log\frac{k}{\ell} + \frac{k}{\ell}\log\frac{k}{\delta_1} \bigg) \label{eq:n_gamma_term} \\
        &= \bigg(\frac{1}{c(\delta_2^{(1)})} \cdot (1-\gamma) k \log k\bigg) (1+o(1)),
\end{align}
since the maximum is achieved by the smallest value $\dmax + 1 = \Theta(k^{\gamma})$, and for that value, the second term is asymptotically negligible compared to the first.  Substituting the definition of $c(\cdot)$ and taking $\epsilon \to 0$, we obtain the condition
\begin{equation}
    n \ge \bigg( \frac{2(1+\frac{1}{3}\delta_2(1-2\rho))}{e^{-\nu}\nu \delta_2^2 (1-2\rho)^2} \bigg) \cdot \Big( (1-\gamma) k \log k\Big) (1+o(1)) \label{eq:conc_final}
\end{equation}

To bound $T_2$, we consider  $\psi_\ell(n,\delta_{2})$ equaling the right-hand side of \eqref{eq:conc_gen_disc} with $\delta = \delta_2^{(2)} I_{\ell}$.  Again upper bounding the summation by $k$ times the maximum, and defining $c'(\delta_2) = \frac{\delta_2^2}{4(8+\delta_2 I_{\ell})}$, we obtain
\begin{equation}
    T_2 \le k\max_{\ell=\frac{k}{\sqrt{\log k}},\dotsc,k}{k \choose \ell} \cdot 2\exp\big(- c'(\delta_2^{(2)}) I_{\ell}^2 n \big).
\end{equation}
It follows that in order to attain $T_2 \le \delta_1$, it suffices that
\begin{equation}
    n \ge \frac{1}{ c'(\delta_2^{(2)}) I_{\ell}^2 } \log\frac{2k\cdot{k \choose \ell}}{\delta_1} \label{eq:n_T2}
\end{equation}
for all $\ell=\frac{k}{\sqrt{\log k}},\dotsc,k$.
By the mutual information characterizations in \eqref{eq:gtn_I_ord}--\eqref{eq:gtn_I_const}, we have $c'(\delta_2^{(2)}) = \Theta(1)$ for any $\delta_2^{(2)} \in (0,1)$, and $I_{\ell}^2 = \Theta\big( \big(\frac{\ell}{k}\big)^2 \big)$.  We consider this fact in the following two cases:
\begin{itemize}
    \item If $\ell = \Theta(k)$, then \eqref{eq:n_T2} simply amounts to $n = \Omega(k)$;
    \item If $\ell = o(k)$, then also writing $\log\frac{2k\cdot{k \choose \ell}}{\delta_1} = \Theta\big( \ell \log \frac{k}{\ell} \big)$, we find that \eqref{eq:n_T2} takes the form $n = \Omega\big( \frac{k^2}{\ell} \log \frac{k}{\ell} \big)$.  The most stringent condition is then provided by the smallest value $\ell = \frac{k}{\sqrt{\log k}}$, yielding $n = \Omega\big( k \cdot \log\log k \cdot \sqrt{\log k} \big)$. 
\end{itemize}
Combining these two cases, we deduce that $T_2$ vanishes for any scaling of the form $n = \Omega\big( k \log \frac{p}{k} \big)$, since $\log \frac{p}{k} = \Theta(\log p) = \Theta(\log k)$ in the sub-linear regime $k = \Theta(p^{\theta})$ with $\theta \in (0,1)$.

\subsubsection{Characterizing the mutual-information based condition \eqref{eq:final_ach}} Recall that we require the number of tests to satisfy \eqref{eq:final_ach}.  For the values of $\ell$ corresponding to $T_1$ in \eqref{eq:T1T2}, we have chosen $\delta_{2,\ell} = \delta_2^{(1)}$, and the mutual information characterization \eqref{eq:gtn_I_ord} yields the condition
\begin{equation}
    n \ge \max_{\ell=\dmax + 1, \dotsc, \frac{k}{\sqrt{\log k}}} \frac{ k\log{\frac{p}{\ell}} + k\log{\frac{k}{\ell}} + \frac{k}{\ell}\log\big(\frac{k}{\delta_1}\big) }{ \big( e^{-\nu}\nu (1-2\rho)\log\frac{1-\rho}{\rho} \big) (1-\delta_{2}^{(1)}) } (1+o(1)), \label{eq:final_ach_T1}
\end{equation}
where we have applied $\log{p-k \choose \ell} = \big(\ell \log{\frac{p}{\ell}}\big)(1+o(1))$ and  $\log{k \choose \ell} = \big(\ell \log{\frac{k}{\ell}}\big)(1+o(1))$ for $\ell = o(k)$.  Writing $k\log{\frac{p}{\ell}} + k\log{\frac{k}{\ell}} = k\log{\frac{p}{k}} + k\log{\frac{k^2}{\ell^2}}$ and recalling that $k = \Theta(p^{\theta})$ and $\dmax = \Theta(k^{\gamma})$, we find that \eqref{eq:final_ach_T1} simplifies to
\begin{equation}
    n \ge \frac{ (1-\theta)k\log p + 2(1-\gamma)\log k }{ \big( e^{-\nu}\nu (1-2\rho)\log\frac{1-\rho}{\rho} \big) (1-\delta_{2}^{(1)}) } (1+o(1)), \label{eq:final_ach_T1a}
\end{equation}
since the maximum over $\ell$ is achieved by the smallest value, $\ell = \dmax + 1 = \Theta(k^{\gamma})$.

For the $\ell$ values corresponding to $T_2$ in \eqref{eq:T1T2}, the condition \eqref{eq:final_ach} was already simplified in \cite{Sca15b}.  It was shown that under the choice $\nu = \log 2$, the dominant condition is that of the highest value, $\ell = k$, and the resulting condition on the number of tests is
\begin{equation}
    n \ge \frac{k\log\frac{p}{k}}{(\log 2 - H_2(\rho)) (1-\delta_2^{(2)})} (1+o(1)). \label{eq:MI1_final}
\end{equation}

\subsubsection{Wrapping up} 

We obtain the final condition on $n$ by combining \eqref{eq:conc_final}, \eqref{eq:final_ach_T1a}, and \eqref{eq:MI1_final}.  We take $\delta_2^{(2)}$ to be arbitrarily small, while renaming $\delta_2^{(1)}$ to $\delta_2$ and letting it remain a free parameter.  Also recalling the choice $\nu = \log 2$, we obtain the following generalization of the partial recovery bound given in \cite{Sca15b}.

\begin{thm}
    Under the symmetric noise model \eqref{eq:gt_symm_model}, in the regime $k = \Theta(p^{\theta})$ and $\dmax = \Theta(k^{\gamma})$ with $\theta,\gamma \in (0,1)$, there exists a non-adaptive group testing algorithm such that $\pe \to 0$ as $p \to \infty$ with a number of tests satisfying
    \begin{equation}
        n \le \inf_{\delta_2 \in (0,1)} \max\big\{ \nMIi, \nMIii(\gamma,\delta_2), \nConc(\gamma,\delta_2) \big\} (1+o(1)),
    \end{equation}
    where $\nMIi$, $\nMIii$, and $\nConc$ are defined in \eqref{eq:nMI1}--\eqref{eq:nConc}.
\end{thm}

\medskip

{\bf Variation for the Z-channel.} For general $\gamma \in (0,1)$, the preceding analysis is non-trivial to extend to the Z-channel noise model, which we consider in Section \ref{sec:asymm}.  However, it is relatively easy to obtain a partial recovery result for the case $\dmax = \Theta(k)$, and such a result suffices for our purposes.  We outline the required changes here.  We continue to assume that the test matrix $\Xv$ is i.i.d.~Bernoulli, but now the probability of a given entry being one is $\frac{\nu}{k}$ for some $\nu > 0$ to be chosen later.

As was observed in \cite{Sca15b}, the analysis is considerably simplified by the fact that we do not need to consider the case $\frac{\ell}{k} \to 0$.  This means that we can rely exclusively on \eqref{eq:conc_gen_disc}, which is known to hold for {\em any} binary-output noise model \cite{Sca15b}.  Consequently, one finds that the only requirement on $n$ is that \eqref{eq:final_ach} holds, with the conditional mutual information $I_{\ell} = I(X_{\sdif};Y|X_{\seq})$ suitably modified due to the different noise model.  By some asymptotic simplifications and the fact that $\ell = \Theta(k)$ for all $\ell$ under consideration, this condition simplifies to
\begin{equation}
    n \ge \max_{\ell > \dmax} \frac{ \ell \log \frac{p}{k} }{ I_{\ell} } (1+o(1)). \label{eq:final_ach_Z}
\end{equation}
Next, we note that an early result of Malyutov and Mateev \cite{Mal80} (see also \cite{Mal13}) implies that $\frac{\ell}{I_{\ell}}$ is maximized at $\ell = k$.  For completeness, we provide a short proof.  Assuming without loss of generality that $s = \{1,\dotsc,k\}$, and letting $X_{j}^{j'}$ denote the collection $(X_{j},\dotsc,X_{j'})$ for indices $1 \le j \le j' \le k$, we have
\begin{align}
    \frac{I_{\ell}}{\ell}
        &= \frac{1}{\ell} I(X_{k-\ell+1}^{k};Y|X_{1}^{k-\ell}) \label{eq:ratio1} \\
        &= \frac{1}{\ell} \sum_{j=k-\ell+1}^k I( X_j; Y | X_{1}^{j-1} ) \label{eq:ratio2} \\
        &= \frac{1}{\ell} \sum_{j=k-\ell+1}^k \big( H(X_j) - H(X_j|Y,X_{1}^{j-1}) \big), \label{eq:ratio3}
\end{align}
where \eqref{eq:ratio1} follows since $I_{\ell} = I(X_{\sdif};Y|X_{\seq})$ only depends on the sets $(\sdif,\seq)$ through their cardinalities, \eqref{eq:ratio2} follows from the chain rule for mutual information, and \eqref{eq:ratio3} follows since $X_j$ is independent of $X_{1}^{j-1}$.  We establish the desired claim by observing that $\frac{I_{\ell}}{\ell}$ is decreasing in $\ell$: The term $H(X_j)$ is the same for all $j$, whereas the term $H(X_j|Y,X_{1}^{j-1})$ is smaller for higher values of $j$ because conditioning reduces entropy.

Using this observation, the condition in \eqref{eq:final_ach_Z} simplifies to
\begin{equation}
    n \ge \max_{\ell > \dmax} \frac{ k \log \frac{p}{k} }{ I_{k} } (1+o(1)). \label{eq:final_ach_Z2}
\end{equation}
We can further replace $I_k = I(X_s; Y)$ by the capacity of the Z-channel upon optimizing the i.i.d.~Bernoulli parameter $\nu > 0$.  The optimal value is the one that makes $\PP[\vee_{j \in s} X_s = 1]$ the same as $P_U^*(1)$, where $P_U^*$ is the capacity-achieving input distribution of the Z-channel $P_{Y|U}$.\footnote{In fact, this analysis applies to any binary channel $P_{Y|U}$.}


\subsection{Partial Recovery Result with Unknown $k$} \label{app:partial_prob}

In this section, we explain how to adapt the partial recovery analysis of \cite{Sca15b} for the symmetric noise model (as well as that of Appendix \ref{app:partial_gamma} for $\dmax = \Theta(k^{\gamma})$) to the case that $k$ is only known to lie within a certain interval $\Kc$ of length $\Delta = o(\dmax)$, where $\dmax$ is the partial recovery threshold.  Specifically, we argue that for any defective set $s$ with $|s| \in \Kc$, there exists a decoder that knows $\Kc$ but not $|s|$, such that the error probability $\PP[\Shat \ne s \,|\, S=s]$ vanishes under i.i.d.~Bernoulli testing, with the same requirement on $n$ is the case of known $|s|$.  Of course, this also implies that $\PP[\Shat \ne S]$ vanishes under any prior distribution on $S$ such that $|S| \in \Kc$ almost surely.

We consider the same non-adaptive setup of Appendix \ref{app:partial_gamma}, denoting the test matrix by $\Xv \in \{0,1\}^p$ and making extensive use of the information densities defined in in \eqref{eq:idens_bn}--\eqref{eq:idens_b}.  Since $k := |s|$ is unknown, we can no longer assume that the test matrix is i.i.d.~with distribution $P_X \sim \Bernoulli\big( \frac{\nu}{k} \big)$, so we instead use $P_X \sim \Bernoulli\big( \frac{\nu}{\kmax} \big)$, with $\kmax$ equaling the maximum value in $\Kc$.

In the case of known $k$, we considered the decoder in \eqref{eq:dec}, first proposed in \cite{Sca15b}.  In the present setting, we modify the decoder to consider all possible $k$, and to allow $\sdif \cup \seq$ to be a strict subset of $s$.  More specifically, the decoder is defined as follows. For any pair $(\sdif,\seq)$ such that $|\sdif \cup \seq|$ equals some constant $k'$, let $\imath_{k'}^n(\xv_{\sdif};\yv|\xv_{\seq})$ be the information density corresponding to the case that the defective set equals $\sdif \cup \seq$, with an explicit dependence on the cardinality $k'$.  We consider a decoder that searches over all $s \subseteq \{1,\dotsc,p\}$ whose cardinality is in $\Kc$, and seeks a set such that
\begin{equation}
    \imath_{k'}^n(\Xv_{\sdif};\Yv|\Xv_{\seq}) \ge \gamma_{k',\ell}, \quad \forall (\sdif,\seq) \subseteq \tilde{\Sc}_s, \label{eq:modified_dec}
\end{equation}
where $\{\gamma_{k',\ell}\}$ is a set of constants depending on $k' := |\sdif \cup \seq|$ and $\ell := |\sdif|$, and $\tilde{\Sc}_s$ is the set of pairs $(\sdif,\seq)$ satisfying the following:
\begin{enumerate}
    \item $\sdif \subseteq s$ and $\seq \subseteq s$ are disjoint;
    \item The total cardinality $k' = |\sdif \cup \seq|$ lies in $\Kc$;
    \item The ``distance'' $\ell + k - k'$ exceeds $\dmax$.  Specifically, if $s$ is the true defective set and $\hat{s}$ is some estimate of cardinality $k' \le k$ with $s \cap \hat{s} = \seq$ and $|\seq| = k' - \ell$, then we have $\ell + k - k'$ false negatives, and $\ell$ false positives, so that $d(s,\hat{s}) = \ell + k - k'$ under the distance function in \eqref{eq:dist}.
\end{enumerate}
If multiple $s$ satisfy \eqref{eq:modified_dec}, then the one with the smallest cardinality $k := |s|$ is chosen, with any remaining ties broken arbitrarily.  If none of the $s$ satisfy \eqref{eq:modified_dec}, an error is declared.

Under this decoder, an error occurs if the true defective set $s$ fails the threshold test \eqref{eq:modified_dec}, or if some $s'$ with $|s'| \le |s|$ and $d(s,s') > \dmax$ passes it.  By the union bound, the first of these occurs with probability at most
\begin{equation}
    \pe^{(1)}(s,\dmax) \le \sum_{\substack{(k',\ell) \,:\, k' \in \Kc, \ell \le k' \le k, \\ \ell+k-k' > \dmax}} {k \choose k'} {k' \choose \ell} \PP\big[ \imath_{k'}^n(\Xv_{\sdif};\Yv|\Xv_{\seq}) \ge \gamma_{k',\ell} \big], \label{eq:pe1}
\end{equation}
where $(\sdif,\seq)$ is an arbitrary pair with $|\sdif| = \ell$ and $|\sdif \cup \seq| = k'$.  Here the combinatorial terms arise by choosing $k'$ elements of $s$ to form $\sdif \cup \seq$, and then choosing $\ell$ of those elements to form $\sdif$.

As for the probability of some incorrect $s'$ passing the threshold test, we have the following.  Let $\seqbar = s' \cap s$ and $\sdifbar = s' \setminus s$.  Since only sets with $|s'| \le |s|$ can cause errors, $k' := |s'| = |\seq \cup \sdifbar|$ is upper bounded by $k$, and since only sets with $d(s,s') > \dmax$ can cause errors, we can also assume that this holds.  Defining $\ell = |\sdif|$, we can upper bound the probability of $s'$ passing the test \eqref{eq:modified_dec} for all $(\sdif,\seq)$ by the probability of passing it for the specific pair $(\sdifbar,\seqbar)$.  By doing so, and summing over all possible $s'$, we find that the second error event is upper bounded as follows for any given $s$:
\begin{equation}
    \pe^{(2)}(s,\dmax) \le \sum_{\substack{(k',\ell) \,:\, k' \in \Kc, \ell \le k' \le k, \\ \ell+k-k' > \dmax}} {p-k \choose \ell} {k \choose k'-\ell} \PP\big[ \imath_{k'}^n(\Xv_{\sdifbar};\Yv|\Xv_{\seqbar}) \ge \gamma_{k',\ell} \big], \label{eq:pe2}
\end{equation}
where the combinatorial terms corresponding to choosing $\ell$ elements of $\{1,\dotsc,p\} \setminus s$ to form $\sdifbar$, and choosing $k' - \ell$ elements of $s$ to form $\seqbar$.

Combining the above, the overall upper bound on the error probability given $s$ is
\begin{equation}    
    \pe(s) \le \pe^{(1)}(s,\dmax) + \pe^{(2)}(s,\dmax). \label{eq:pe_total}
\end{equation}
Upon substituting the upper bounds in \eqref{eq:pe1} and \eqref{eq:pe2}, we obtain an expression that is nearly the same as that when $k$ is known \cite{Sca15b}, except that we sum over a number of different $k'$, rather than only $k' = k$.  We proceed by arguing that this does not affect the final bound, as long as $\dmax = \Theta(k^{\gamma})$ for some $\gamma \in (0,1]$, and $\Delta = o(\dmax)$ (recall that $\Delta$ is the highest possible difference between two $k$ values).

The main additional difficulty here is that the information density $\imath_{k'}(x_{\sdif};y|x_{\seq}) = \log\frac{P_{Y|X_{\sdif},X_{\seq}}(y|x_{\sdif},x_{\seq})}{P_{Y|X_{\seq}}(y|x_{\seq})}$ is defined with respect to $(\seq,\sdif)$ of total cardinality $k'$, whereas the output variables $\Yv$ are distributed according to the true model in which there are $k$ defectives.  The following lemma allows us to perform a {\em change of measure} to circumvent this issue. 

\begin{lem} \label{lem:pr_ratio}
    Fix a defective set $s$ of cardinality $k$,  let $(\sdif,\seq)$ be disjoint subsets of $s$ with total cardinality $k' \le k$, and let $P^{(k)}_{Y|X_{\sdif},X_{\seq}}$ be the conditional probability of $Y$ given the partial test vector $(X_{\sdif},X_{\seq})$, in the case of a test vector with i.i.d.~$\mathrm{Bernoulli}\big(\frac{\nu}{\kmax}\big)$ entries, where $\kmax = k(1+o(1))$.  Similarly, let $P^{(k')}_{Y|X_{\sdif},X_{\seq}}$ denote the conditional transition law when $s' = \sdif \cup \seq$ is the true defective set.  Then, if $|k - k'| \le \Delta = o(k)$, we have
    \begin{equation}
        \max_{x_{\sdif},x_{\seq},y} \frac{ P^{(k)}_{Y|X_{\sdif},X_{\seq}}(y|x_{\sdif},x_{\seq}) }{ P^{(k')}_{Y|X_{\sdif},X_{\seq}}(y|x_{\sdif},x_{\seq}) } \le 1 + O\bigg( \frac{\Delta}{k} \bigg). \label{eq:pr_ratio}
    \end{equation}
    Consequently, the corresponding $n$-letter product distributions $P^{(k)}_{\Yv|\Xv_{\sdif},\Xv_{\seq}}$ and $P^{(k')}_{\Yv|\Xv_{\sdif},\Xv_{\seq}}$ for conditionally independent observations satisfy the following:
    \begin{equation}
        \max_{\xv_{\sdif},\xv_{\seq},\yv} \frac{ P^{(k)}_{\Yv|\Xv_{\sdif},\Xv_{\seq}}(\yv|\xv_{\sdif},\xv_{\seq}) }{ P^{(k')}_{\Yv|\Xv_{\sdif},\Xv_{\seq}}(\yv|\xv_{\sdif},\xv_{\seq}) } \le e^{O(\frac{n\Delta}{k})} \label{eq:pr_ratio_n}
    \end{equation}
\end{lem}
\begin{proof}
    First observe that if $x_{\sdif}$ or $x_{\seq}$ contain an entry equal to one, then the ratio in \eqref{eq:pr_ratio} equals one, as $Y = 1$ with probability $1 - \rho$ in either case.  Hence, it suffices to prove the claim for $x_{\sdif}$ and $x_{\seq}$ having all entries equal to zero.  In the denominator, we have
    \begin{equation}
        P^{(k')}_{Y|X_{\sdif},X_{\seq}}(1|x_{\sdif},x_{\seq}) = \rho,  \label{eq:pk'_eval}   
    \end{equation}
    since there $(\sdif,\seq)$ corresponds to the entire defective set.  On the other hand, in the numerator, there are $k - k'$ additional defective items, and the probability of one or more of them being defective is $\epsilon := 1 - \big( 1 - \frac{\nu}{k} )^{k - k'} = O\big( \frac{\Delta}{\kmax} \big)$, where we applied the assumptions $|k - k'| \le \Delta = o(k)$ and $\kmax = k(1+o(1))$, along with some asymptotic simplifications.  Therefore, we have
    \begin{align}
        P^{(k)}_{Y|X_{\sdif},X_{\seq}}(1|x_{\sdif},x_{\seq}) 
            &= (1-\epsilon)\rho + \epsilon(1-\rho) \\    
            &= \rho + \epsilon(1-2\rho). \label{eq:pk_eval}
    \end{align}
    The ratio of \eqref{eq:pk_eval} and \eqref{eq:pk'_eval} evaluates to $1 + O(\epsilon)$, and similarly for the conditional probabilities of $Y=0$ obtained by taking one minus the right-hand sides.  Since $\epsilon = O\big( \frac{\Delta}{k} \big)$, this proves \eqref{eq:pr_ratio}.
    
    We obtain \eqref{eq:pr_ratio_n} by raising the right-hand side of \eqref{eq:pr_ratio} to the power of $n$, and applying $1 + \alpha \le e^{\alpha}$.
\end{proof}

We now show how to use Lemma \ref{lem:pr_ratio} to bound $\pe^{(1)}(s,\dmax)$ and $\pe^{(2)}(s,\dmax)$.  Starting with the former, we observe that $\Yv$ in \eqref{eq:pe1} is conditionally distributed according to $P^{(k)}_{\Yv|\Xv_{\sdif},\Xv_{\seq}}$, and hence, \eqref{eq:pr_ratio_n} yields
\begin{equation}
    \PP\big[ \imath_{k'}^n(\Xv_{\sdif};\Yv|\Xv_{\seq}) \ge \gamma_{k',\ell} \big] \le e^{O(\frac{n\Delta}{k})} \cdot \PP\big[ \imath_{k'}^n(\Xv_{\sdif};\widetilde{\Yv}|\Xv_{\seq}) \ge \gamma_{k',\ell} \big], \label{eq:apply_ratio_1}
\end{equation}
where $\widetilde{\Yv}$ is conditionally distributed according to $P^{(k')}_{\Yv|\Xv_{\sdif},\Xv_{\seq}}$.

For $\pe^{(2)}(s,\dmax)$, we first note that a similar bound to \eqref{eq:pr_ratio_n} holds when we condition on $\Xv_{\seq}$ alone; this is seen by simply moving the denominator to the right-hand side and averaging over $\Xv_{\sdif}$ on both sides.  Since $\Yv$ in \eqref{eq:pe2} is conditionally distributed according to $P^{(k)}_{\Yv|\Xv_{\seqbar}}$, we obtain from \eqref{eq:pr_ratio_n} that
\begin{equation}
\PP\big[ \imath_{k'}^n(\Xv_{\sdifbar};\Yv|\Xv_{\seqbar}) \ge \gamma_{k',\ell} \big] \le e^{O(\frac{n\Delta}{k})} \cdot \PP\big[ \imath_{k'}^n(\Xv_{\sdifbar};\widetilde{\Yv}|\Xv_{\seqbar}) \ge \gamma_{k',\ell} \big], \label{eq:apply_ratio_2}
\end{equation}
where $\widetilde{\Yv}$ is conditionally distributed according to $P^{(k')}_{\Yv|\Xv_{\seqbar}}$.

Next, observe that if the number of tests satisfies $n = O(k \log p)$, then we can simplify the term $e^{O(\frac{n\Delta}{k})}$ to $e^{O(\Delta \log p)}$.  By doing so, and substituting \eqref{eq:apply_ratio_1} and \eqref{eq:apply_ratio_2} into \eqref{eq:pe1}--\eqref{eq:pe_total}, we obtain
\begin{align}
    \pe(s) &\le e^{O(\Delta \log p)} \sum_{\substack{(k',\ell) \,:\, k' \in \Kc, \ell \le k' \le k, \\ \ell+k-k' > \dmax}} {k \choose k'} {k' \choose \ell}  \PP\big[ \imath_{k'}^n(\Xv_{\sdif};\widetilde{\Yv}|\Xv_{\seq}) \ge \gamma_{k',\ell} \big] \nonumber \\
        & + e^{O(\Delta \log p)}\sum_{\substack{(k',\ell) \,:\, k' \in \Kc, \ell \le k' \le k, \\ \ell+k-k' > \dmax}} {p-k \choose \ell} {k \choose k'-\ell} \PP\big[ \imath_{k'}^n(\Xv_{\sdifbar};\widetilde{\Yv}|\Xv_{\seqbar}) \ge \gamma_{k',\ell} \big].
\end{align}
This bound is now of a similar form to that analyzed in \cite{Sca15b}, in the sense that the joint distributions of the tests and outcomes match those that define the information density.  The only differences are the presence of additional $k'$ values beyond only $k = k'$, and the presence of the $e^{O(\Delta \log p)}$ terms.  We conclude by explaining how these differences do not impact the final result as long as $\Delta = o(\dmax)$ with $\dmax = \Theta(k^{\gamma})$ for some $\gamma \in (0,1]$:
\begin{itemize}
    \item The term ${p - k \choose \ell}$ satisfies $\log {p -k \choose \ell} = \big(\ell\log\frac{p}{\ell}\big)(1+o(1))$, and the assumption $|k - k'| \le \Delta = o(\dmax) = o(k)$ implies that the term ${k' \choose \ell}$ satisfies $\log {k' \choose \ell} = \big(\ell\log\frac{k}{\ell}\big)(1+o(1))$.  On the other hand, the logarithm of  $e^{O(\Delta \log p)}$ is $O(\Delta \log p)$, so it is dominated by the other combinatorial terms due to the fact that $\Delta = o(\dmax)$ and $\ell = \Omega(\dmax)$.  Similarly, the term ${k \choose k'} = {k \choose k-k'}$ satisfies $\log {k \choose k'} = O(\Delta \log k)$, and is dominated by ${k' \choose \ell}$.
    \item The term ${k \choose k' - \ell}$ simplifies to ${k \choose k - k' + \ell} = {k \choose \ell(1+o(1))}$ (by the assumption $\Delta = o(\dmax)$), and hence, the asymptotic behavior for any $k'$ is the same as ${k \choose k - \ell}$, the term corresponding to $k = k'$.  Similarly, the asymptotics of the tail probabilities of the information densities are unaffected by switching from $k$ to $k' = k(1+o(1))$.
    \item In \cite{Sca15b}, the number of $\ell$ being summed over is upper bounded by $k$, whereas here we can upper bound the number of $(k',\ell)$ being summed over by $k\Delta$.  Since $\Delta = o(k)$, this simplifies to $k^{1+o(1)}$.  Since it is the logarithm of this term that appears in the final expression, this difference only amounts to a multiplication by $1+o(1)$.
\end{itemize}

\subsection{NCOMP with Unknown Number of Defectives} \label{app:ncomp}

Chan {\em et al.} \cite{Cha11} showed that Noisy Combinatorial Orthogonal Matching Pursuit (NCOMP), used in conjunction with i.i.d.~Bernoulli test matrices, ensures exact recovery of a defective set $S$ of cardinality $k$ with high probability under the scaling $n = O( k\log p )$, which in turn behaves as $O\big( k\log\frac{p}{k} \big )$ when $k = O(p^{\theta})$ for some $\theta < 1$.  However, the random test design and the decoding rule in \cite{Cha11} assume knowledge of $k$, meaning the result cannot immediately be used for our purposes in Step 2 of Algorithm \ref{alg:steps}.  In this section, we modify the algorithm and analysis of \cite{Cha11} to handle the case that $k$ is only known up to a constant factor.

Suppose that $k \in [c_0\kmax,\kmax]$ for some $\kmax = \Theta(p^{\theta})$, where $c_0 \in (0,1)$ and $\theta \in (0,1)$ do not depend on $p$.  We adopt a Bernoulli design in which each item is independently placed in each test with probability $\frac{\nu}{\kmax}$ for fixed $\nu > 0$.  It follows that for a given test vector $X = (X_1,\dotsc,X_p)$, we have
\begin{equation}
    \PP\bigg[\bigvee_{j \in S} X_j = 1 \bigg] = 1 - \bigg(1 - \frac{\nu}{\kmax}\bigg)^{k} = (1 - e^{-c\nu}) (1+o(1))
\end{equation}
for some $c \in [c_0,1]$, and hence, the corresponding observation $Y$ satisfies
\begin{equation}
    \PP[Y = 1] = \Big( (1-\rho)(1 - e^{-c\nu}) + \rho e^{-c\nu}\Big) (1+o(1)). \label{eq:ncomp_p_uncond}
\end{equation}
In contrast, for any $j \in S$, we have
\begin{equation}
    \PP[Y = 1 | X_j = 1] = 1 - \rho. \label{eq:ncomp_p_cond}
\end{equation}
The idea of the NCOMP algorithm is the following: For each item $j$, consider the set of tests in which the item is included, and define the total number as $N'_j$.  If $j$ is defective, we should expect a proportion of roughly $1-\rho$ of these tests to be positive according to \eqref{eq:ncomp_p_cond}, whereas if $j$ is non-defective, we should expect the proportion to be roughly $ (1-\rho)(1 - e^{-c\nu}) + \rho e^{-c\nu}$ according to \eqref{eq:ncomp_p_uncond}.  Hence, we set a threshold in between these two values, and declare $j$ to be defective if and only if the proportion of positive tests exceeds that threshold.

We first study the behavior of $N'_j$.  Under the above Bernoulli test design, we have $N'_j \sim \mathrm{Binomial}\big(n,\frac{\nu}{\kmax}\big)$, and hence, standard Binomial concentration \cite[Ch.~4]{Mot10} gives
\begin{align}
    \PP\bigg[ N'_j \le \frac{n\nu}{2\kmax} \bigg] 
        &\le e^{-\Theta(1) \frac{n}{\kmax}} \\
        &\le \frac{1}{p^2}, \label{eq:N'_conc2}
\end{align}
where \eqref{eq:N'_conc2} holds provided that $n = \Omega(k \log p)$ with a suitably-chosen implied constant (recall that $k = \Theta(\kmax)$).

Next, we present the modified NCOMP decoding rule, and study its performance under the assumption that $N'_j = n'_j$ with $n'_j \ge \frac{n\nu}{2\kmax}$, for each $j \in \{1,\dotsc,p\}$.  Observe that the gap between \eqref{eq:ncomp_p_uncond} and \eqref{eq:ncomp_p_cond} behaves as $\Theta(1)$ for any $c \in [c_0,1]$.  Hence, for sufficiently small $\Delta > 0$, we have $\PP[Y=1] \le 1 - \rho - 2\Delta$.  Accordingly, letting $N'_{j,1}$ be the number of the $N'_j$ tests including $j$ that returned positive, we declare $j$ to be defective if and only if $N'_{j,1} \ge (1 - \rho - \Delta)N'_j$.  We then have the following:
\begin{itemize}
    \item If $j$ is defective, then the probability of incorrectly declaring it to be non-defective given $N'_j = n'_j$ satisfies
    \begin{equation}
        \PP\big[ N'_{j,1} < (1 - \rho - \Delta)n'_j \big] \le e^{-\Theta(1) n'_j} \le e^{-\Theta(1) \frac{n\nu}{2\kmax}},
    \end{equation}
    where the first inequality is standard Binomial concentration, and the second holds for $n'_j \ge \frac{n\nu}{2\kmax}$.
    \item Similarly, if $j$ is non-defective, the probability of incorrectly declaring it to be defective given $N'_j = n'_j$ satisfies
    \begin{equation}
        \PP\big[ N'_{j,1} \ge (1 - \rho - \Delta)n'_j \big] \le e^{-\Theta(1) n'_j} \le e^{-\Theta(1) \frac{n\nu}{2\kmax}}.
    \end{equation}
\end{itemize}
Combining these bounds with \eqref{eq:N'_conc2} and a union bound over the $p$ items, the overall error probability $\pe = \PP[\Shat\ne S]$ of the modified NCOMP algorithm is upper bounded by
\begin{equation}
    \pe \le \frac{1}{p} + p e^{-\Theta(1) \frac{n\nu}{2\kmax}}.
\end{equation}
Since $\kmax = \Theta(k)$, this vanishes when $n = \Omega(k \log p)$ with a suitably-chosen implied constant, thus establishing the desired result.

{\bf Z-channel noise.}  Under the Z-channel noise model introduced in Section \ref{sec:asymm}, the preceding analysis is essentially unchanged.  It only relied on there being a constant gap between the probabilities $\PP[Y = 1]$ and $\PP[Y = 1 \,|\, X_j = 1]$, and this is still the case here: Equations \eqref{eq:ncomp_p_uncond} and \eqref{eq:ncomp_p_cond} remain true when $(1-\rho)(1 - e^{-c\nu}) + \rho e^{-c\nu}$ is replaced by $(1 - e^{-c\nu}) + \rho e^{-c\nu}$ in the former.

\section*{Acknowledgment}

The author thanks Volkan Cevher, Sidharth Jaggi, Oliver Johnson, and Matthew Aldridge for helpful discussions, and Leonardo Baldassini for sharing his PhD thesis \cite{BalThesis}.  This work was supported by an NUS startup grant.

 \bibliographystyle{IEEEtran}
 \bibliography{JS_References}
 
\end{document}

%% file: preamble.tex
\usepackage[T1]{fontenc}
\usepackage[latin9]{inputenc}
\usepackage{amsthm}
\usepackage{amsmath}
\usepackage{amssymb}
\usepackage{graphicx}
\usepackage{dsfont}
\usepackage{cite}
\usepackage{amsmath}
\usepackage{array}
\usepackage{multirow}
\usepackage{caption}
\usepackage{color}

\usepackage{multicol}

\theoremstyle{plain}

\theoremstyle{plain}

\theoremstyle{plain}
\newtheorem{lem}{\protect\lemmaname}
\theoremstyle{plain}
\newtheorem{thm}{\protect\theoremname}
\theoremstyle{plain}
  
\theoremstyle{definition}

\theoremstyle{definition}

\theoremstyle{definition}

\makeatother
  
\usepackage{babel} 

\providecommand{\claimname}{Claim}
\providecommand{\lemmaname}{Lemma}
\providecommand{\propositionname}{Proposition}
\providecommand{\theoremname}{Theorem}
\providecommand{\corollaryname}{Corollary} 
\providecommand{\definitionname}{Definition}
\providecommand{\assumptionname}{Assumption}
\providecommand{\remarkname}{Remark}

\newcommand{\overbar}[1]{\mkern 1.25mu\overline{\mkern-1.25mu#1\mkern-0.25mu}\mkern 0.25mu}

\newcommand{\openone}{\mathds{1}}

\newcommand{\Shat}{\widehat{S}}

\newcommand{\nMIi}{n_{\mathrm{MI,1}}}
\newcommand{\nMIii}{n_{\mathrm{MI,2}}}
\newcommand{\nConc}{n_{\mathrm{Conc}}}
\newcommand{\nIndiv}{n_{\mathrm{Indiv}}}

\newcommand{\kmax}{k_{\mathrm{max}}}

\newcommand{\ncheck}{\check{n}}

\newcommand{\ntil}{\widetilde{n}}
\newcommand{\Ntil}{\widetilde{N}}
\newcommand{\Ncheck}{\check{N}}

\newcommand{\sbar}{\overbar{s}}

\newcommand{\sdif}{s_{\mathrm{dif}}}
\newcommand{\sdifbar}{\overbar{s}_{\mathrm{dif}}}
\newcommand{\seqbar}{\overbar{s}_{\mathrm{eq}}}
\newcommand{\seq}{s_{\mathrm{eq}}}

\newcommand{\dmax}{d_{\mathrm{max}}}

\newcommand{\Bernoulli}{\mathrm{Bernoulli}}

\newcommand{\pe}{P_{\mathrm{e}}}

\newcommand{\xv}{\mathbf{x}}

\newcommand{\Xv}{\mathbf{X}}
\newcommand{\yv}{\mathbf{y}}
\newcommand{\Yv}{\mathbf{Y}}

\newcommand{\Ac}{\mathcal{A}}

\newcommand{\Kc}{\mathcal{K}}

\newcommand{\Sc}{\mathcal{S}}

\newcommand{\Yc}{\mathcal{Y}}

\newcommand{\EE}{\mathbb{E}}
\newcommand{\PP}{\mathbb{P}}